\documentclass[11pt,A4]{article}
\usepackage[hidelinks]{hyperref} 
\usepackage[margin= 2cm]{geometry}
\usepackage[T1]{fontenc}
\usepackage[utf8]{inputenc}
\usepackage{amsthm}
\usepackage{thm-restate}
\usepackage{todonotes,wrapfig,float,graphicx,amssymb,textcomp,array,amsmath}
\graphicspath{{./Figures}}
\usepackage{enumerate,enumitem,tipa}
\usepackage{mathtools}
\usepackage{subcaption}
\usepackage{accents}
\usepackage{multirow}
\usepackage{tabularx}
\usepackage{color,xcolor}
\usepackage{comment}
\usepackage[normalem]{ulem}

\usepackage{lineno,tikz}
\usepackage{cleveref}
\newcommand{\old}[1]{{}}
\newcommand{\later}[1]{{}}

\hypersetup{
  colorlinks=true,
  linkcolor=red!70!black,
  linktoc=all,
  citecolor=blue,
    pdfpagemode=FullScreen,
}

\definecolor{Darkblue}{rgb}{0,0,.8}
\definecolor{Brown}{cmyk}{0,0.61,1.,0.60}
\definecolor{Purple}{cmyk}{0.45,0.86,0,0}
\definecolor{Darkgreen}{rgb}{0.133,0.700,0.133}

\definecolor{MyGreen}{rgb}{0.200,0.500,0.200}

\newtheorem{theorem}{Theorem}
\newtheorem{problem}{Problem}

\newtheorem{claim}[theorem]{Claim}
\newtheorem{claim*}[theorem]{Claim}
\newtheorem{corollary}[theorem]{Corollary}

\newtheorem{definition}[theorem]{Definition}
\newtheorem{observation}[theorem]{Observation}

\newenvironment{claimproof}[1][\proofname]{\proof[#1]}{\endproof}

\newcommand{\cupdot}{\mathbin{\mathaccent\cdot\cup}}

\newcommand{\IR}{\mathbb{R}}
\newcommand{\eps}{\varepsilon}
\newcommand{\Tz}{\mathsf{Top}}
\newcommand{\Tnz}{\mathsf{Side}}

\newcommand{\IH}{\mathbb{H}}
\newcommand{\Hyp}{\mathbb{H}}
\newcommand{\IZ}{\mathbb{Z}}
\newcommand\dist{\ensuremath{\mathrm{dist}}}
\newcommand\poly{\ensuremath{\mathrm{poly}}}

\newcommand\diam{\ensuremath{\mathrm{diam}}}
\newcommand\conv{\ensuremath{\mathrm{conv}}}

\newcommand\arsinh{\ensuremath{\mathrm{arsinh}}}
\newcommand\grid{\ensuremath{\mathrm{grid}}}
\newcommand\qt{Q_{\bf a}(P)}

\newcommand{\zd}{z_{\min}}
\newcommand{\zu}{z_{\max}}
\newcommand{\x}{x_{\min}}
\newcommand{\xm}{x_{\max}}
\newcommand{\xd}{x_{\min}}

\newcommand{\wt}{\mathrm{len}}
\newcommand{\bd}{\partial}

\newcommand{\opt}{\textsf{OPT}}

\newcommand{\EMPH}[1]{\emph{\textcolor{red!50!black}{#1}}}

\newtagform{gray}{\color{gray!70!black}(}{)}

\newcommand{\arc}[1]{%
  \tikz[baseline={(0,-0.5ex)}]{
    \node[inner sep=0pt] (a) {$#1$};
    \draw (a.north west) .. controls +(0,1ex) and +(0,1ex) .. (a.north east);
  }%
}

\begin{document}
\usetagform{gray}

\title{Gap-ETH-Tight Algorithms for Hyperbolic TSP and Steiner Tree\thanks{This work was supported by the Research Council of Finland, Grant 363444.}}
\author{
Sándor Kisfaludi-Bak\thanks{Email: \texttt{sandor.kisfaludi-bak@aalto.fi}}
\and
Saeed Odak\thanks{Email: \texttt{saeed.odak@aalto.fi}}
\and
Satyam Singh\thanks{Email: \texttt{satyam.singh@aalto.fi}}
\and
Geert van Wordragen\thanks{Email: \texttt{geert.vanwordragen@aalto.fi}}
}

\date{Department of Computer Science, Aalto University, Espoo, Finland.}

\maketitle       

\begin{abstract}
We give an approximation scheme for the TSP in $d$-dimensional hyperbolic space that has optimal dependence on $\varepsilon$ under Gap-ETH. For any fixed dimension $d\geq 2$ and for any $\varepsilon>0$ our randomized algorithm gives a $(1+\varepsilon)$-approximation in time $2^{O(1/\varepsilon^{d-1})}n^{1+o(1)}$. We also provide an algorithm for the hyperbolic Steiner tree problem with the same running time.

Our algorithm is an Arora-style dynamic program based on a randomly shifted hierarchical decomposition. However, we introduce a new hierarchical decomposition called the hybrid hyperbolic quadtree to achieve the desired large-scale structure, which deviates significantly from the recently proposed hyperbolic quadtree of Kisfaludi-Bak and Van Wordragen (JoCG'25). Moreover, we have a new non-uniform portal placement, and our structure theorem employs a new weighted crossing analysis. We believe that these techniques could form the basis for further developments in geometric optimization in curved spaces.
\end{abstract}

\section{Introduction}  \label{sec:intro}

The metric traveling salesman problem (TSP) is easily stated: given a set of $n$ points in some metric space, find a cyclic permutation of points that minimizes the sum of lengths between consecutive points. A very commonly studied variant is when the underlying space is $d$-dimensional Euclidean space (henceforth denoted by $\IR^d$), and admits $(1+\eps)$-approximations in polynomial time (PTAS) for any fixed $\eps>0$~\cite{Arora98,Mitchell99}. The PTAS of Arora~\cite{Arora98} for Euclidean TSP has been tirelessly improved, generalized, and optimized throughout the past decades by researchers working on geometric optimization. Today, we have Gap-ETH-tight running times in Euclidean space, and have working approximation schemes in any so-called doubling space to many related problems~\cite{BartalG13,BartalG21,BartalGK16Siam,BorradaileKM15,EucTSPJACM,Moemke25,RaoS98}.
Moreover, many components of an Arora-style scheme can also be used in other important optimization problems related to network design, clustering, and facility location~\cite{AroraRR98,CohenAddad20,CohenAddadFS21,CohenAddadKM19,CzumajJKV22,FakcharoenpholW25,WuFW24}.

The key components of Arora's scheme~\cite{Arora98} involve a hierarchical decomposition of the space (a quadtree) where the space is decomposed into fat\footnote{An object (a point set) is \EMPH{$\alpha$-fat} if the radius of the minimum circumscribed ball divided by the radius of the maximum inscribed ball is at least $\alpha$. It is \EMPH{fat} if it is $\alpha$-fat for some fixed constant $\alpha>0$.} regions called \EMPH{cells}, each of which in turn is decomposed into $f(d)$ smaller fat cells, where $d$ is the dimension of the space. One can then argue that using small modifications, called \EMPH{patching} on the tour, we can find a $(1+\eps)$-approximate tour that crosses between neighboring cells of the decomposition only $g(\eps,d)$ times (or, in some cases, $f(\eps,d)\cdot \poly(\log n)$ times), and these crossings come from a fixed set of points called \EMPH{portals}.
However, this argument cannot be made for a fixed hierarchical decomposition; one needs to apply a randomized decomposition and ensure that the number of crossings is small in expectation. Finally, one can employ a bottom-up dynamic programming algorithm on the hierarchy to compute partial solutions for each possible connection pattern of the cell's portals. Crucially, the running time depends on the number of portals used by the approximate tour as well as the number of children a cell can have, as both of these terms appear in the exponent of the dynamic programming algorithm.

In the Euclidean case, the hierarchical decomposition is a quadtree whose nodes correspond to axis-aligned hypercubes whose side length is a power of $2$. The entire quadtree is shifted randomly, ensuring that in expectation the hypercube boundaries avoid hitting many short segments of the tour.

Doubling spaces are those metric spaces where for all $r\geq 0$ it holds that all balls of radius $r$ can be covered by constantly many balls of radius $r/2$. The above techniques have been successfully generalized to this much more general setting, where randomized net-trees and padded decompositions take over the role of randomly shifted quadtrees~\cite{BuschCFHHR23, ChanGMZ16,Filtser19,GuptaKL03,La025,Talwar04}.

However, generalizing beyond the doubling space setting has been difficult. A natural target of extensions into the non-doubling setting is hyperbolic space, where the local structure is (almost) Euclidean, but at super-constant scale the space expands exponentially. Moreover, the exponential expansion leads to a tree-like behavior that can be exploited~algorithmically.

Geometric optimization in hyperbolic space is still in its infancy, but there is reason to believe that even low-dimensional hyperbolic space is highly relevant from the practical perspective~\cite{BlasiusFKMPW22,BlasiusFK16,BringmannKL19,FriedrichK18,GaneaBH18,HyperSteiner,MnasriJS15,NickelK18,ShavittT04,ungar2013einstein,VerbeekS14}, and establishing the basic tools of geometric optimization in this setting will be widely useful for future algorithmic developments in both theory and practice. Thus, it would be prudent to explore the following problem.

\begin{problem}
    Adapt algorithms from Euclidean geometric optimization to negatively curved (non-doubling) spaces, or prove lower bounds that rule out such adaptations.
\end{problem}

Krauthgamer and Lee~\cite{KrauthgamerL06} showed the existence of a PTAS for TSP in a more general setting: in so-called geodesic, visual and Gromov-hyperbolic spaces. This implies the existence of a PTAS in $d$-dimensional hyperbolic space (henceforth denoted by $\Hyp^d$); however, the running time dependencies on $d$ and $\eps$ are not explicit in their work. Their paper relies on doubling space results that bring factors $1/\eps^{O(d)}$ to the exponent of the running time, thus Gap-ETH-tight running times cannot result from their approach. (Nonetheless, we include a more direct overview of their techniques and achieved running times to ours, see \Cref{sec:overview}.)

In order to obtain tight running times, we need to make a more explicit analysis and not rely on generic doubling space results.

\begin{problem}\label{probTSP}
Is it possible to solve TSP in $\Hyp^d$ as fast as in $\IR^d$?    
\end{problem}

A natural starting point is the recent work on hierarchical decompositions in hyperbolic space and the newly introduced hyperbolic quadtree of Kisfaludi-Bak and Van Wordragen~\cite{Kisfaludi-BakW24}. In fact, Garc\'ia-Castellanos, Medbouhi, Marchetti, Bekkers, and Kragic~\cite{HyperSteiner} recently raised the question if a quadtree-based PTAS for hyperbolic Steiner tree is possible. We answer both the question of \Cref{probTSP} and of \cite{HyperSteiner} affirmatively.

\subsection{Our results}
In this paper, we show that explicit Arora-style approximation schemes are possible in hyperbolic space, and we can match the efficiency of Euclidean approximation schemes apart from a factor of $\poly(\log n)$.

\begin{restatable}{theorem}{thmmain}\label{thm_main}
    For any fixed $d \geq 2$ and any $\eps>0$, there is a randomized $(1 + \eps)$-approximation for TSP and Steiner tree in $d$-dimensional hyperbolic space in $2^{O(1/\eps^{d-1})}n(\log n)^{2d(d-1)}$ time.
\end{restatable}

To make our algorithm work, we need to introduce three new key components compared to earlier work in hyperbolic and Euclidean spaces:
\begin{itemize}
    \item We need to change the hyperbolic quadtree of Kisfaludi-Bak and Van Wordragen~\cite{Kisfaludi-BakW24} into a so-called hybrid hyperbolic quadtree (in short, hybrid tree). We create a hyperbolic shifting technique.
    \item We introduce a new portal placement for the cells of this hybrid tree. The portal placement is non-uniform, in stark contrast with all earlier algorithms.
    \item The proof of our structure theorem is built on a new amortized \emph{weighted} crossing analysis.
\end{itemize}

Our techniques readily generalize to the Steiner tree problem and achieve the same running time. This resolves a recent question posed in~\cite{HyperSteiner} about hyperbolic Steiner tree. Our result can be considered as an important step toward geometric approximation schemes in non-doubling spaces of negative curvature, and hopefully a first step toward geometric optimization in more general (non-doubling) Riemannian manifolds.

It is important to note that due to our modified shifting technique, there is no polynomial-time derandomization for this algorithm for general hyperbolic point sets, again deviating from the Euclidean setting. However, derandomization is possible in case the input has at most constant diameter, and it multiplies the running time by $n^d$ in this case. Our derandomization can be stated as follows.

\begin{restatable}{corollary}{corrdeterministic}
    For any fixed $d \geq 2$ and any $\eps>0$, there is a deterministic $(1+\eps)$-approximation for TSP and Steiner tree in $\IH^d$ in $2^{O(\diam(P) + 1/\eps^{d-1})} n^{d+1}$ $(\log n)^{2d(d-1)}$ time, where $\diam(P)$ is the diameter of the input point set $P \subset \IH^d$.
\end{restatable}

The achieved dependence on $\eps$ is identical to the fastest Euclidean algorithm~\cite{EucTSPJACM,Moemke25}, which is best possible under the so-called Gap Exponential Time Hypothesis (Gap-ETH)~\cite{Dinur16,ManurangsiR17}. Since a instance in $\IR^d$ can be scaled to a small neighborhood and embedded in $\Hyp^d$ with infinitesimally small distortion, any lower bound in $\IR^d$ readily applies in $\Hyp^d$.

\begin{corollary}[Embedding the lower bound construction of \cite{EucTSPJACM}]
    For any fixed integer $d\geq 2$ there is a $\gamma>0$ such that there is no $2^{\gamma/\eps^{d-1}}\poly(n)$-time $(1+\varepsilon)$-approximation algorithm for TSP in $\Hyp^d$ under Gap-ETH.
\end{corollary}

In particular, this means that our running time for TSP is Gap-ETH-tight for any constant $d$. However, it remains an open problem whether a \emph{deterministic} running time of $2^{O(1/\eps^{d-1})}n^{O(1)}$ is possible in $\Hyp^d$ for point sets of super-constant diameter.

\subsection{Overview of our techniques, comparison with earlier work}\label{sec:overview}

\subparagraph{The techniques of Krauthgamer and Lee~\cite{KrauthgamerL06}.}

The main idea in the TSP algorithm of Krauthgamer and Lee~\cite{KrauthgamerL06}  is to use a randomized embedding of the space into a distribution of trees. This exploits the tree-likeness of hyperbolic spaces at larger scales. At the same time, they propose an Arora-style approach inside balls of radius $\tau$. Roughly, the space is partitioned into ball-like regions of radius $\tau$ in a randomized fashion, then embedded into a tree whose nodes correspond to partition classes.

The value of $\tau$ and the obtained running times are not explicit in~\cite{KrauthgamerL06}. Here we will briefly speculate about the resulting running time for the sake of comparison.
When embedding into a tree with edges of length at least $\tau=\Omega_\delta(1/\eps)$ one can observe that the tree distance and the true hyperbolic distance will differ by at most $O(\delta)=O(1)$, which is at most $(1+O(\eps))\tau$, i.e., measuring distances between different classes along the tree leads to only $1+O(\eps)$ multiplicative error.

Within the partition classes of diameter $\tau$, Krauthgamer and Lee must apply doubling space or Arora-style techniques~\cite{KrauthgamerL06}. However, even with newer doubling space algorithms, e.g.~\cite{BartalGK16Siam,Talwar04}, the doubling dimension of this neighborhood appears in the second level exponent. Since the doubling dimension of a $\tau$-radius ball is \emph{exponential in $\tau$}, this leads to an at least triple-exponential algorithm in $1/\eps$. Similarly, using some version of Arora's algorithm~\cite{Arora98} on a distorted Euclidean instance would have to have the solution quality depend on the distortion of the $\tau$-neighborhood, which is again exponential in $\tau$. Thus, Arora's algorithm would have to be invoked with precision at least $\eps'=2^{-\tau}$, leading to a running time that is at least double-exponential in $\tau=\Omega(1/\eps)$.

\subparagraph{Overcoming exponential expansion.}

In a more explicit Arora-style algorithm, we have some challenges to overcome, in particular the following.
\begin{description}
    \item[Ch1] Euclidean quadtrees (and net-trees in doubling spaces) have growth rates that depend only on the dimension, i.e., each cell has $f(d)$ children cells for some fixed function $f$. This is important as the dynamic program must take into account all portal placements of all children, and in particular, the running time has the number of children in the exponent. Due to the exponential expansion of hyperbolic spaces, no hierarchical decomposition matching the following properties is possible: (a) cells are fat, (b) cells of level $\ell$ have diameter $2^{\Theta(\ell)}$ and (c) each cell is partitioned into $f(d)$ children cells for some fixed function $f$. This issue was already observed by~\cite{Kisfaludi-BakW24}.
\end{description}
This property means that the hyperbolic quadtree of~\cite{Kisfaludi-BakW24} is not suitable for our purposes: the number of children a cell can have is unbounded. Instead, we propose a new structure, called \emph{hybrid hyperbolic quadtree}, or just \emph{hybrid tree} for short, that has ``open'' cells at positive levels, and it fails the diameter property, but keeps fatness, and the number of children for cells is now bounded by $2^d$.
At the same time, we can still use the quadtree of~\cite{Kisfaludi-BakW24} for its negative levels; these resemble a slightly distorted $d$-dimensional Euclidean quadtree, which is only defined now up to cells of diameter $O(1)$.

\begin{figure}
    \centering
    \includegraphics[page=2,width=\textwidth]{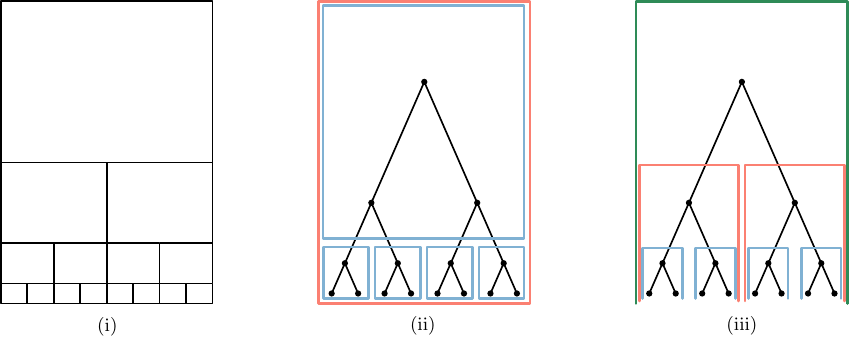}
    \caption{(i) The binary tiling of $\Hyp^2$ with isometric tiles. (ii) The cells of various positive levels of the corresponding quadtree illustrated on the binary tree $T$. (iii) Cells of the new hybrid tree of various positive levels illustrated on $T$.}
    \label{fig:hybridtree}
\end{figure}

The binary tiling~\cite{Boroczky74,CannonFKP97} is a binary tree-like decomposition of the hyperbolic plane into isometric fat tiles. The tiling has a structure that resembles a binary tree: below each tile there are two other tiles, see the illustration in the half-plane model of $\Hyp^2$ in \Cref{fig:hybridtree}. For the sake of a simpler description, we compare the quadtree of~\cite{Kisfaludi-BakW24} and our hybrid tree via this binary tree $T$. In the quadtree of~\cite{Kisfaludi-BakW24}, a cell of level $\ell$ will correspond to a subtree of $T$ of height $2^\ell$, which can naturally be cut at half its depth into trees of height $2^{\ell-1}$, all the way until the tiling at level $0$ is obtained consisting of subtrees of height $1$ (i.e., the vertices of $T$). In our hybrid tree, a so-called \emph{cell} of level $\ell>0$ simply corresponds to a subtree of height~$\ell$. Here we only allow subtrees whose leaves are also leaves of $T$. The corresponding cell diameter is not exponential in $h$, but each cell naturally has $2$ children that are cells, plus a single child that is a cell of the binary tiling.

Notice that the resulting cells are `open' from the leaves' side, i.e., there is no cell below these cells. In particular, we place no portals on the bottom sides of cells of level $\ell\geq 1$.

\subparagraph*{Our algorithm.}
Handling the hybrid tree in a black box fashion, our algorithm proceeds similarly to Arora's scheme.
We first round our input points to centers of cells of size $\eps\cdot \opt/n$. The level of these cells is $\ell_{\min}$, and it is useful to think of $\ell_{\min}<0$ for this section.

We apply a random shift on the hybrid tree using the Euclidean vectors in the half-plane model. The shifting is very similar to the Euclidean shifting, with horizontal shifts corresponding to vertical cell boundaries in the bottom layer of the lowest level cells. On the vertical axis, the shifting has only constant range. While the shifting is easy to describe and somewhat natural, it is significantly different from the grid-based shifting in Euclidean space.

The main difficulty, as with all algorithms of this type, is the structure theorem: we need to prove that there exists a $(1+\eps)$-approximate tour that intersects the cell boundaries only at a few pre-defined portals. In the Euclidean setting, one can prove that the number of intersections between the optimum tour and all grid hyperplanes is $O(\opt)$ (assuming a minimum edge length of one), and then one can show that the patching cost incurred at each of these intersection points is $O(\eps)$ in expectation. Since hyperbolic spaces are not vector spaces, most of the steps of the analysis fail and need new ideas: we already fail to bound the number of crossings in the same manner, as there is no convenient collection of hyperplanes that could play the role of grid hyperplanes in the Euclidean analysis.

\subparagraph{The number of crossings, weighted portal analysis.} 
In Arora's algorithm~\cite{Arora98}, grid lines parallel to each coordinate axis are analyzed in the same uniform manner. However, in a binary tiling, the vertical lines can intersect cells at larger heights and have a non-uniform behaviour. The main challenges with vertical lines (and in higher dimensions, vertical hyperplanes) are as follows:
\begin{description}
    \item[Ch2] The vertical lines corresponding to our shifts can intersect the optimum tour an unbounded number of times, while in the Euclidean setting, the number of intersections with vertical grid lines is (up to constant factors) equal to the tour length.
    
    We overcome this via weighing the intersection points by the inverse of their $z$-coordinates. This allows us to have an absolute bound (for every shift) on the number of intersections between the tour and the vertical lines as a function of the length of the optimum tour.
    
    \item[Ch3] In dimension $d\geq 3$, a side facet of a cell is itself a cell whose structure resembles that of a binary tree. In particular, its surface area is an exponential function of its diameter. This is a very significant problem, as all known algorithms rely on placing portals in a grid-like fashion, ensuring that the distance from a grid point to the nearest portal is at most $\eps$ times the diameter of the facet. In Euclidean space, this requires the placement of $O(1/\eps^{d-1})$ portals inside the hypercube-shaped facet. Here, this would lead to placing a number of portals that is exponential in $1/\eps$, thus it would lead to a double-exponential running time as a function of $1/\eps$. Since the number of portals is in the exponent of the running time of Arora-style algorithms, this would lead to a double-exponential dependence on $(1/\eps)$.

    We overcome this challenge via an elaborate portal placement. First, we observe that the probability that an intersection point ends up in a very large cell facet is exponentially small. Thus, we can afford to pay a large patching cost for certain crossings near the bottom of the bounding box with an exponentially small probability. Moreover, we must ensure that the patching cost gets progressively smaller as we get closer to the top of the side facets, as we need to counterbalance the weighted analysis used for the number of crossings. This leads to a non-uniform portal placement. For an illustration of the portal placement on a 2-dimensional $\Tnz$ facet of a 3-dimensional cell of the hybrid tree, see~\Cref{fig:portals}(ii).
\end{description}   
\begin{figure}
    \centering
    \includegraphics[width=0.85\textwidth]{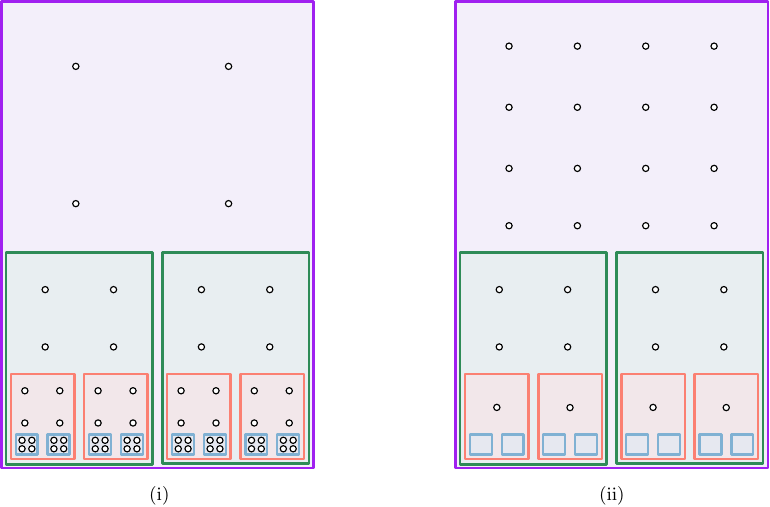}
    \caption{Illustration of (i) uniform (na\"{i}ve) portal placement and  (ii) non-uniform portal placement, on a $\Tnz$ facet of a 3-dimensional cell of the hybrid tree. For simplicity, Figure (ii) uses a scaling factor of $2$, although the true scaling factor is $2^{1-1/d}$. The empty circles represent the portals.}
    \label{fig:portals}
\end{figure}

\subparagraph{The dynamic programming: a banyan for hyperbolic Steiner tree.}

The dynamic programming algorithm for hyperbolic TSP can proceed in a bottom-up manner on the hybrid tree, as seen in previous algorithms~\cite{Arora98,EucTSPJACM}. In fact, the same structure theorem readily works for the Steiner tree problem, however, there is one difficulty: in the leaf cells (and in compressed cells) of the hybrid tree, the part of the optimum falling there is some Steiner forest, and such a forest can have Steiner points inside these cells. It is non-trivial to find these Steiner points in such a way as to preserve the approximation ratio. In the Euclidean setting, the standard solution is to use a so-called \emph{banyan} of the portals of the cell.

A \EMPH{$(1+\eps)$-banyan} of the point set $P$ is a geometric graph $G$ on some set $P\cup Q$ such that for any subset $S\subseteq P$ there is a Steiner tree in $G$ for $S$ whose length is at most $1+\eps$ times longer than the optimum Steiner tree of $S$ in the ambient space. In the Euclidean setting of leaf cells, one can take $Q$ to be a dense grid and $G$ to be a complete geometric graph on $P\cup Q$: this will lead to a solution of size $\poly(1/\eps)$. However, due to the exponential volume growth in the hyperbolic setting, this is not a viable solution in hyperbolic space at larger distance scales: in this case, we show that it suffices to place potential Steiner points along all edges of a ($d$-dimensional) triangulation of portals and input points.

\subsection{Organization of the paper}
We start by defining the hybrid tree in \Cref{Compressed-Quadtree}, as well as the shifts and portal placement.
Based on this, we state our structure theorem and sketch its proof in \Cref{sec_structure}.
Our approximation schemes follow from the structure theorem, mostly using standard techniques; we go into this in \Cref{sec_DP}. We also present our hyperbolic banyan in \Cref{sec_DP}. We finish with a conclusion in \Cref{conclusion}.

\section{Preliminaries}\label{sec_prelims}
Let $\IH^d$ and $\IR^d$ denote $d$-dimensional hyperbolic space (of sectional curvature $-1$) and $d$-dimensional Euclidean space, respectively, with distance functions $\dist_\IH$ and $\dist_\IR$.
For a given set of points $P=\{p_1,p_2,\ldots,p_n\}$, a tour is defined to be a cycle $\pi$ that visits each point and returns to its starting point. Let $\wt_\IR(\pi)$ denote its total Euclidean length and $\wt_{\IH}(\pi)$ its total hyperbolic length.
In the \EMPH{hyperbolic Traveling Salesman problem}, the goal is to determine the hyperbolic length $\opt_{\IH}$ of the shortest such tour passing through $P$.
We also let $\opt_{\IR}$ denote the Euclidean length of an optimal Euclidean tour. 
A \EMPH{hyperbolic Steiner tree} of a point set $P \subseteq \IH^d$ is a connected geometric graph whose vertex set includes~$P$. The aim of the \EMPH{hyperbolic Steiner Tree problem} is to compute a Steiner tree of $P$ with minimum total hyperbolic length.
Throughout this paper, we use $\log$ for the base-2 logarithm and $\ln$ for the natural logarithm. For a positive integer $n$, we use $[n]$ to denote the set $\{1,\dots,n\}$.
For $X \subset \mathbb{H}^d$, we write \EMPH{$\diam(X)$} for the largest hyperbolic distance between any two points~of~$X$. For $\delta>0$, a set $N$ is a \EMPH{$\delta$-cover} of~$X$ if every $x \in X$ lies within distance $\delta$ of some $y \in N$.

\subparagraph{Poincaré upper half-space model.}
The \EMPH{Poincaré upper half-space model} maps hyperbolic space $\IH^d$ to the Euclidean upper half-space, where each point  $p\in\IH^d$ is represented as $(x(p), z(p)) \in \IR^{d-1} \times \IR_{>0}$.
Although the model maps points into Euclidean space, the hyperbolic distance between two points $p$ and $q$ is not given by their Euclidean distance, and the shortest path between $p$ and $q$ is typically not a Euclidean straight line segment. In this model, the geodesic (shortest path) between $p$ and $q$ is either a vertical Euclidean line segment (when $x(p)=x(q)$), or a circular arc from a Euclidean circle that is perpendicular to the Euclidean hyperplane $z = 0$. The hyperbolic distance between $p$ and $q$ is given as
\[
\dist_{\IH}\big((x(p),z(p)),(x(q),z(q))\big)
= 2\arsinh\left(\sqrt{\frac{\lVert x(p)-x(q)\rVert_2^2+(z(p)-z(q))^2}{4z(p)z(q)}}\right),
\] 
and when $x(p)=x(q)$, this simplifies to $\left|\ln \frac{z(p)}{z(q)}\right|$.
Similar to vertical Euclidean lines, vertical Euclidean hyperplanes are also hyperbolic hyperplanes.
However, horizontal Euclidean hyperplanes are ``horospheres''. For example, they can be intersected by the same line segment twice.
The results in this paper are presented in the upper half-space model, (i.e., our input point set is given via its coordinates in this model). However, the results independent of the chosen model, since all the standard models of hyperbolic space are equivalent, and one can convert between them with high precision in linear time. In the half-space model we think of the $z$-direction as going \emph{upwards}. 
For more details, see~\cite{CannonFKP97} or the textbooks~\cite{benedetti1992lectures,iversen1992hyperbolic,thurston97three}.

\subparagraph{Hyperbolic quadtree of~\cite{Kisfaludi-BakW24}.}
In the upper half-space model, an \EMPH{axis-aligned horobox} $B$ is a Euclidean axis-aligned box specified by two corner points: $(x_{\min}(B), \zd(B))$ and $(x_{\max}(B), \zu(B))$, where each coordinate of the first point is lexicographically minimal and each coordinate of the second point is lexicographically maximal. Such a horobox is bounded by two horospheres and $2(d-1)$ vertical hyperplanes.
A \EMPH{cube-based horobox} $B$ is an axis-aligned horobox with $z_{\max}(B) / z_{\min}(B) = 2^h$ and $(x_{\max}(B) - x_{\min}(B)) / z_{\min}(B) = (w, \dots, w)$, where $w = w(B)$ is called the \EMPH{width} of $B$ and $h= h(B)$ is called the \EMPH{height} of $B$.

In contrast to the Euclidean quadtree, which is based on a grid, the hyperbolic quadtree is constructed from the binary tiling. To define the binary tiling, we begin with the cube-based horobox $C$ such that $\xd(C)=(0,\cdots,0)$ and $\zd=1$, $w(C)=1/\sqrt{d-1}$ and $h(C)=1$.
For $\sigma > 0$ and $\tau \in \IR^{d-1}$, let $T_{\sigma, \tau}(x, z) = \sigma \cdot (x  +\tau, z)$ be a family of transformations. 
The \EMPH{binary tiling} is obtained by applying transformations $T_{\sigma, \tau}$, where $\sigma = 2^k$ for integers $k$, and $\tau \in \IZ^{d-1}$, to the horobox $C$.
This yields a tiling of $\IH^d$ into isometric cube-based horoboxes; see \Cref{fig:hybridtree}(i).

A \EMPH{hyperbolic quadtree} is a hierarchical decomposition of $\IH^d$
based on the binary tiling introduced in~\cite{Kisfaludi-BakW24}.
Let $T_\ell$ denote the tiling of $\IH^d$ by cube-based horoboxes at level
$\ell \in \IZ$, where $T_0$ is the \emph{standard binary tiling}.
The cells across different levels are related through the hierarchical structure
of the binary tiling as follows.
\begin{itemize}
    \item \textbf{Levels $\ell \ge 0$.}
    For every $\ell \ge 0$, each cell of $T_{\ell+1}$ is the union of a cell
    $C \in T_\ell$ together with its $2^{d-1}$ children,
    i.e., the cells of $T_\ell$ directly below $C$ in the binary tiling
    (see \Cref{fig:binarytiling}).
    \item \textbf{Levels $\ell < 0$.}
    For every $\ell < 0$, the tiling $T_\ell$ is obtained by applying
    the standard Euclidean quadtree split to each cell of $T_{\ell+1}$,
    producing $2^d$ children per parent.
\end{itemize}
Level~$0$ serves as the transition point:
the quadtree is hyperbolic (binary tiling–based) for all $\ell \ge 0$
and Euclidean for all $\ell < 0$.

\begin{figure}
    \centering
    \hspace{2.5cm}\includegraphics[width=0.45\linewidth]{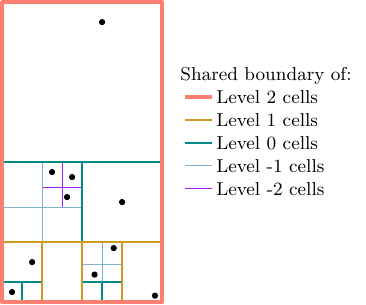}
    \caption{Illustration of the hyperbolic quadtree.}
    \label{fig:binarytiling}
\end{figure}

\subparagraph*{Euclidean tools from \cite{Arora98} and \cite{EucTSPJACM}.}
For small distances, our structure theorem relies on the Euclidean structure theorems of \cite{Arora98} and \cite{EucTSPJACM}. We recall the necessary definitions and the formal statement of these Euclidean structure theorems.

Let $P \subset \IR^d$ be contained in the hypercube $[0,L]^d$.  
Choose a random shift ${\bf a} = (a_1,\dots,a_d)$, where each $a_i \in \{1,\dots,L\}$.  
Consider the hypercube  
$C({\bf a}) = \prod_{i=1}^d [-a_i + \frac{1}{2},\, 2L - a_i + \frac{1}{2}]$.  
Note that $[0,L]^d \subseteq C({\bf a})$ and hence $P \subseteq C({\bf a})$.  
Let $D({\bf a})$ be the Euclidean quadtree under the shift ${\bf a}$, with $C({\bf a})$ as its root node.
We say that two distinct cells of $D({\bf a})$ with the same side length are \EMPH{siblings} if they share the same parent cell.

The following lemma (known as \emph{patching lemma}) is generally used to reduce the number of ways a tour or Steiner tree can cross a given hyperplane. For the proof, one can refer to~\cite{Arora98,EucTSPJACM}.
\begin{theorem}[Patching Lemma~\cite{Arora98}]\label{arora_patching}
There is a constant $c > 0$ such that the following holds. Let $\ell$ be an arbitrary line segment and $\pi$ be a tour (or Steiner tree) that crosses $\ell$ at least thrice. Then there exist line segments on $\ell$ whose total length is at most $c\cdot \wt_{\IR}(\ell)$ and whose addition to $\pi$ changes it into a closed path (resp., Steiner tree) $\pi'$ that crosses $\ell$ at most twice.
\end{theorem}

\subparagraph{Building blocks of the technique of Arora~\cite{Arora98}.} We now briefly describe the building blocks from~\cite{Arora98} that we will use. An \EMPH{$m$-regular set} of portals on a $d$-dimensional hypercube $C$ is an orthogonal lattice \EMPH{$\grid(C,m)$} of $m$ points in $C$. If the cube has side length $\ell$, then the spacing between the portals is set to $\ell/(m^{1/d}-1)$. We say that a set of line segments $S$ is \EMPH{$r$-light} with respect to the dissection $D({\bf a})$ if it crosses each face of each cell of $D({\bf a})$ at most $r$ times.

\begin{theorem}[Arora’s Structure Theorem~\cite{Arora98}]\label{thm:arorastruct}
Let $P \subseteq \{0, \ldots, L\}^d$.
Let $\bf a$ be a random shift. Then, with probability at least $1/2$, there is a salesman tour or Steiner tree of cost at most $(1+\eps)\opt_{\IR}$ that is $r$-light with respect to $D({\bf a})$ such that it crosses each facet $F$ of a cell of $D({\bf a})$ only at points from $\grid(F,m)$, for some $m = (O((d/\eps) \log L))^{d-1}$ and $r = (O(d/\eps))^{d-1}$. 
\end{theorem}

\subparagraph{Building blocks of the technique of Kisfaludi-Bak, Nederlof and W\k{e}grzycki~\cite{EucTSPJACM}.}
 We now briefly describe the building blocks from~\cite{EucTSPJACM} that we will use. For a $(d-1)$-dimensional hypercube $F$, let $F^*$ denote $F$ without its $2^{d-1}$ corner points.
\begin{definition}[Euclidean~$r$-simple geometric graph]\label{def:rsimple}
Let $\pi$ be a geometric graph such that the grid hyperplanes of $D({\bf a})$ do not contain any edge of $\pi$. We say that $\pi$ is \EMPH{$r$-simple} if for every facet $F$ shared by a pair of sibling cells in $D({\bf a})$, we have
\begin{itemize}
    \item $\pi$ crosses $F^*$ through at most one point (and some subset of $2^{d-1}$ corners of $F$), or
    \item  $\pi$ crosses $F$ only through the points from $\grid(F,g)$ for some $2^{d-1} \le g \le r^{2d-2}\ |\ \pi \cap F^* |$.
\end{itemize}
Moreover, for any point $p$ on a hyperplane $h$ of $D({\bf a})$, $\pi$ crosses $h$ at most twice via $p$.
\end{definition}

\begin{definition}[Euclidean~$r$-simplification]\label{def:rsimplification}
We say that a geometric graph $\pi'$ is an \EMPH{$r$-simplification} of $\pi$ if $\pi'$ is \emph{$r$-simple}, and in each facet $F$ where $\lvert \pi \cap F^* \rvert = 1$, the single non-corner crossing is a point from $\pi \cap F^*$.
\end{definition}

\begin{theorem}[Kisfaludi-Bak, Nederlof and W\k{e}grzycki\ Structure Theorem~\cite{EucTSPJACM}] \label{thm:sparsitysensEuclid}
Let $\bf a$ be a random shift and let $\pi$ be a tour or Steiner tree of $P\subseteq \IR^d$. Then, for any positive integer $r$ there is a tour (resp. Steiner tree) $\pi'$ of $P$ that is an $r$-simplification of $\pi$ such that
  $$ \mathbb{E}_{\mathbf{a}}[\wt_{\IR}(\pi') - \wt_{\IR}(\pi)]\leq O{\left(d^{5/2}\cdot \frac{\wt_{\IR}(\pi)}{r}\right)}.$$
\end{theorem}

We note the following observation regarding the proof of \Cref{thm:arorastruct} by~\cite{Arora98} and the proof of \Cref{thm:sparsitysensEuclid} by~\cite{EucTSPJACM}.

\begin{restatable}{observation}{euclideanstructuretheorem}\label{obs:structproofEuc}
    There is a set of points $X$ in the intersection of each grid hyperplane $H$ with the root cell $C_0$ of the Euclidean shifted quadtree such that the proof modifies a tour or Steiner tree $\pi$ using a set $S$ of segments in $H\cap C_0$ whose total expected length is $O(d^{5/2} / r) \cdot \wt_{\IR}(\pi)$. Moreover, if the intersections on some hyperplanes are not patched (i.e., we only do a partial patching), then the cost of the patching does not increase.
\end{restatable}

\section{Hybrid Hyperbolic Quadtree and Portals} \label{Compressed-Quadtree}

The \EMPH{hybrid hyperbolic quadtree} is a tree whose nodes, called cells, are cube-based horoboxes. Each cell is the disjoint union of its children in the tree and has an integer level such that same-level cells are (almost) isometric.
A cell of level $0$ is a cell of the binary tiling.
A cell of level $\ell > 0$ is the union of its children, which are a cell $C$ from the binary tiling and the $2^{d-1}$ cells of level $\ell - 1$ directly below $C$ (for an illustration in $\IH^2$, see~\Cref{fig:hybridtree}(iii)). 
 Equivalently, a cell $C$ of level $\ell \geq 0$ is a portion of the binary tiling resembling a complete $2^{d-1}$-ary tree of height~$\ell$, whose subtrees of height $\ell' > 0$ are exactly the level-$\ell'$ descendants of $C$ (see \Cref{fig:hybridtree}(iii)). Note that all the cells of the same positive level are isometric.
Cells of levels $\ell < 0$ are defined by splitting a binary tiling cell using Euclidean quadtree splits as in the hyperbolic quadtree of \cite{Kisfaludi-BakW24}.

We can also define a \EMPH{compressed hybrid hyperbolic quadtree} (shown in \Cref{fig:comp_tree}) analogously to the compressed Euclidean quadtree.
A \EMPH{compressed hybrid hyperbolic quadtree} on $P$ is obtained from hybrid hyperbolic quadtree $Q(P)$ by contracting every maximal path of degree-one nodes. Every node $v$ stores its cell $C_v$ and its level $\ell(v)$. If a maximal degree-one path starts at $v$ and ends at its unique child, called $C_{\text{child}}$, we keep $v$ and $C_{\text{child}}$, make $C_{\text{child}}$ the only child of $v$, call $v$ a \EMPH{compressed cell}, and the cell that corresponds to $v$ is $C_v \setminus C_{\text{child}}$. 
If $v$ is not compressed, the cell that corresponds to $v$ is just $C_v$.
The only nodes with one child are (possibly trivial) compressed cells, and each can be charged to one of the at most $|P|-1$ branching nodes. Hence, the compressed hybrid hyperbolic quadtree has $O(2^{d} |P|)$ nodes. For an illustration of compressed cells and a compressed hybrid hyperbolic quadtree, see~\Cref{fig:comp_tree}.

From this point onward, we refer to the hybrid hyperbolic quadtree as the \EMPH{hybrid tree} and to the compressed hybrid hyperbolic quadtree as the \EMPH{compressed hybrid tree}.

\begin{restatable}{lemma}{lemcompressedhybridtreecomplexity}\label{lem:compressed-hybrid-tree-complexity}
Let $P$ be a set of $n$ points in $\IH^d$. There exists an $O(2^d n + dn\log{n})$ time algorithm to build a compressed hybrid tree on $P$.
\end{restatable}
\begin{proof}
    On the negative levels, for each non-empty cell of the binary tiling we construct a compressed quadtree using \cite{Aluru04} in $O(dn \log n)$ time. Next, we compute the compressed binary tiling $T^0(P)$ following Lemma~4.10 of Kisfaludi-Bak and Van Wordragen~\cite{Kisfaludi-BakvW25} in $O(dn \log{n})$ time (by exploiting that this resembles a $(d-1)$-dimensional Euclidean quadtree). Finally, to obtain the compressed hybrid tree $\qt$ we extend each of the cells of $T^0(P)$ downwards to contain their $2^{d-1}$ child cells in the hybrid tree.
\end{proof}

 \begin{figure}[htbp]
    \centering
    \includegraphics[width=0.95\textwidth]{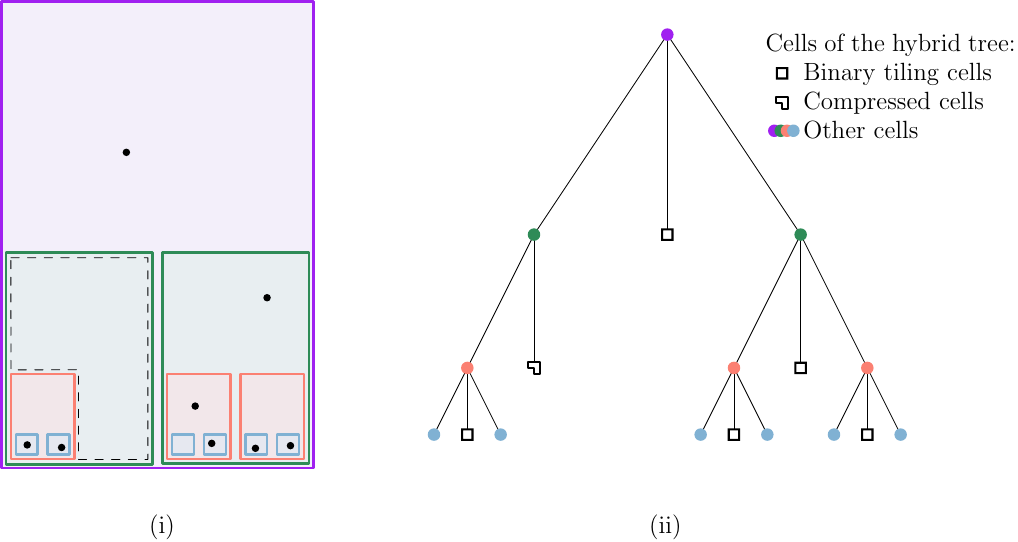}
    \caption{The compressed hybrid tree. (i) A point set and the cells of its compressed hybrid tree (A compressed cell is depicted in dashed lines). (ii) The tree structure of the hybrid tree.}  \label{fig:comp_tree}
\end{figure}

\subparagraph{Bounding box.}
Let $B$ be the smallest Euclidean axis-aligned bounding box containing~$P$.
We first apply the same isometry $T_{\sigma,\tau}$ to both $P$ and $B$ so that $\x(B) = (0,\dots,0)$ and $\zd(B) = 1$.
Let $s$ be the largest of the distances between non-adjacent facets of $B$.
\begin{restatable}{lemma}{clmdiam}\label{clm_diam}
    $\diam(B) \leq (d+1)s$.
\end{restatable}
\begin{proof}
    Take two arbitrary points $p,q \in B$ and let $p',q'$ be their projections onto the top facet of $B$.
    The distance from $p$ to $p'$ is at most the distance from the bottom facet to the top facet of $B$. Thus, we have $\dist(p,p') \leq s$ and $\dist(q,q') \leq s$.
    Note that for any two vertical facets, their minimum distance is attained by points that lie at the top of the facet.
    Thus, we have $\dist(p',q') \leq (d-1)s$. Adding these up gives $\dist(p,q) \leq (d+1)s$, as claimed.
\end{proof}
Since $B$ is the smallest bounding box, $\opt_\IH \geq s$.
Now, due to~\Cref{clm_diam}, $\opt_\IH \geq (1/(d+1))\cdot\diam(B)$. 
Set $\ell_{\min}$ to be the minimum of~$0$ and the largest level where twice the diameter of a cell at that level is less than $(\eps/(d+1)n)\cdot \diam(B)$.
Define $\ell_{\max}-1$ as the smallest hybrid tree level at which some cell $C$ fully contains $B$, then let $B^*$ denote the level $\ell_{\max}$ hybrid tree cell whose lexicographically minimal child cell is $C$. 

\subparagraph{Perturbation.} 
Similar to Arora~\cite{Arora98}, our algorithm perturbs the instance by snapping each point in $P$ to the center of level $\ell_{\min}$ cells (horobox) that contain them. Specifically, each point lying strictly inside a cell is mapped to the center of that cell, while points located on the boundaries of cells are assigned arbitrarily to the center of one of the adjacent cells. After snapping, we introduce an infinitesimal perturbation by adding the offset $\mu\cdot(1,1,\ldots,1)$ to each perturbed point in that cell, for an arbitrarily small $\mu > 0$. This ensures that, under every possible shift of the grid, no perturbed point lies exactly on a cell boundary.

\begin{restatable}{lemma}{clmperturb}\label{clm_perturb}
    The length of an optimal tour over the perturbed points is at most $(1+\eps)\opt_{\IH}$. 
\end{restatable}
\begin{proof}
    Recall that each point in $P$ is moved to the centers of a level $\ell_{\min}$ cell, where the cells have diameter at most $\eps\cdot \opt_{\IH}/(2n)$.
    Since multiple points may map to the same center, the perturbation may reduce the number of distinct locations. Nevertheless,  the distance each point moves is at most the diameter of a cell in $\ell_{\min}$, which is at most $\eps\cdot \opt_{\IH}/(2n)$. Therefore, the additional cost incurred in the tour due to this perturbation is at most $n \cdot (\eps\cdot \opt_{\IH}/n) = \eps\cdot\opt_{\IH}$. Hence, it follows that the total cost of the optimal tour over the perturbed points is at most $\opt_{\IH} + \eps\cdot \opt_{\IH} = (1+\eps)\opt_{\IH}$, as required.
\end{proof}

\subsection{Shifting of the {hybrid tree}} \label{Shifting-Quadtree}
Similarly to Euclidean approximation schemes for TSP, we will use random shifting, though it is worth noting that our ``shift'' will not be an isometry.
Let $\mathcal{C}_{\min}$ denote the cells of level $\ell_{\min}$ that lie in the lexicographically minimal child cell of $B^*$.
We pick ${\bf a}_x \in \{ \x(C) \mid C \in \mathcal{C}_{\min} \}$ and ${\bf a}_z \in \{ \zd(C) \mid C \in \mathcal{C}_{\min} \} \cap [1,2)$ uniformly at random, then apply the transformation $(x,z) \mapsto (x - {\bf a}_x, z / {\bf a}_z)$ to the hybrid tree.
This would be an isometry if ${\bf a}_z$ rescaled all coordinates; however, since it only rescales the $z$-coordinate, the distances are distorted by a factor of at most ${\bf a}_z \in [1,2)$.
Consequently, the transformed compressed hybrid tree has slightly different properties from the original hybrid tree, but we can still talk about cells and their children. 
We denote the shifted compressed hybrid tree under shift ${\bf a} = ({\bf a}_x, {\bf a}_z)$ by $\qt$.
We stop splitting once cells reach level $\ell_{\min}$, even if they still contain multiple points.

\begin{restatable}{lemma}{pointsinleaves}\label{lem_pointsinleaf}
    Any level $\ell_{\min}$ cell of $\qt$ contains at most $2^{d-1}$ distinct perturbed points.
\end{restatable}
\begin{proof}
    Let $C$ be an $\ell_{\min}$ cell in $\qt$. Let $C'$ be the cell of $Q(P)$ that corresponds to $C$ before shifting. Assume $C^{b}$ is the cell of the binary tiling that contains the cell $C'$. Since ${\bf a}_z \in [1, 2)$, the perturbed points inside $C$ are potentially the centers of cells in $Q(P)$ at the same $z$-range either within the cell $C^{b}$ or the $2^{d-1}$ cells of the binary tilling directly below $C^{b}$. In the former case there is at most one point in $C$ and in the later case there are at most $2^{d-1}$ perturbed points in $C$.
\end{proof}

\begin{restatable}{lemma}{lemconstdistort}\label{lem_constdistort}
    Let $C$ be a cell of level $0$ in $\qt$.
    There is a bijection $\phi : C \to [0,1]^d$ such that $\phi(C)$ is a Euclidean unit hypercube that is split in $\qt$ exactly as a Euclidean (compressed) quadtree.
    Additionally, $\phi$ is bi-Lipschitz: for any $p,q \in C$,
    \[\Omega{\left(dist_{\IH}(p,q)\right)} \leq \dist_{\IR}(\phi(p), \phi(q)) \leq O{\left(\sqrt{d}\cdot \dist_{\IH}(p,q)\right)}. \]
\end{restatable}
\begin{proof}
We use the transformation
\[ \phi(x,z) = \left( (x - \x(C)) \cdot \frac{\sqrt{d-1}}{a_z \cdot \zd(C)},\ \log\frac{z}{\zd(C)} \right). \]
One can verify that this indeed maps $(\x(C), \zd(C))$ to $(0,\dots,0)$ and $(\xm(C), \zu(C))$ to $(1,\dots,1)$.
Additionally, the point $\phi^{-1}(\frac{1}{2}, \dots, \frac{1}{2})$ is the common intersection of the Euclidean hyperplanes used to split up $C$.

Next, we start proving the distance property, with $\frac{\dist_{\IR}(p,q)}{\zd(C)}$ as an intermediate step.

\begin{claim}\label{clm_constdistort1}
    Let $a = (x_1,\cdots,x_{d-1}, z)$ and $b = (x_1',\cdots,x_{d-1}', z') \in \IR^d$, where $z, z' > 0$ and $x_i, x_i' \in \IR$, for $i\in[d-1]$. Let $a^*=\phi(a)$ and $b^*=\phi(b)$. Then, we have 
    \[\Omega\left(\frac{\dist_{\IR}(a,b)}{\zd(C)}\right)\leq \dist_{\IR}(a^*,b^*) \leq O\left(\frac{\sqrt{d}\cdot\dist_{\IR}(a,b)}{\zd(C)}\right).\]
\end{claim}

\begin{claimproof}
 Let $\dist_{\max}(a,b)=\max\{\lVert x-x'\rVert_2,|z-z'|\}$.
Since $a^*=\phi(a)$ and $b^*=\phi(b)$, we have 
$$a^*= \left(x\cdot\frac{\sqrt{d-1}}{a_z \cdot\zd(C)},\ \log\frac{z}{\zd(C)}\right),$$
and
$$b^*= \left(x'\cdot\frac{\sqrt{d-1}}{a_z\cdot\zd(C)},\ \log\frac{z'}{\zd(C)}\right).$$
Thus,
\begin{align*}
    \dist_{\max}(a^*,b^*)=&\max\left\{
\tfrac{\sqrt{d-1}}{a_z\cdot\zd(C)}\lVert x-x'\rVert_2,\
\Bigl|\log\tfrac{z}{\zd(C)}-\log\tfrac{z'}{\zd(C)}\Bigr|
\right\}.
\end{align*}
Note that
    if $\zd(C)>0$ and $z,z'\ge \zd(C)$, then we have
 $$\frac{1}{2\zd(C)}|z-z'| \le \left|\log \frac{z}{\zd(C)}-\log \frac{z'}{\zd(C)}\right|\leq \frac{1}{\zd(C)}|z-z'|.$$ \label{eqn_zz'}
Therefore, we have 
\begin{equation*}
   c_2\cdot \dist_{\max}(a^*,b^*)\leq \max\left\{
\tfrac{\sqrt{d-1}}{a_z\cdot\zd(C)}\lVert x-x'\rVert_2,\ 
\frac{|z-z'|}{\zd(C)}
\right\}\leq c_3\cdot \dist_{\max}(a^*,b^*),
\end{equation*}
 where $c_2$ and $c_3$ are positive constants.

Note that
\begin{align}\label{eqn_ub_infty}
c_2\cdot\dist_{\max}(a^*,b^*)
\leq& \max\left\{
\tfrac{\sqrt{d-1}}{a_z\cdot\zd(C)}\lVert x-x'\rVert_2,\ 
\frac{|z-z'|}{\zd(C)}
\right\} \\
\leq& \tfrac{\sqrt{d-1}}{\zd(C)}\cdot\max\left\{\lVert x-x'\rVert_2,\ |z-z'|\right\}\notag\\
=& \tfrac{\sqrt{d-1}}{\zd(C)}\cdot\dist_{\max}(a,b)
\end{align}

Observe that for any $s,t \in \mathbb{R}^d$, we have  \begin{equation}\label{eqn_relation}
    \dist_{\max}(s,t) \le \dist_{\mathbb{R}}(s,t)\le \sqrt{2}\cdot\dist_{\max}(s,t).
\end{equation}
Applying~\eqref{eqn_relation} to $a,b$ and to $a^*,b^*$ in~\eqref{eqn_ub_infty},
\begin{equation}\label{eqn_ub}
    \dist_{\mathbb{R}}(a^*,b^*)\le \sqrt{2}\cdot\dist_{\max}(a^*,b^*)\le\frac{\sqrt{2}}{c_2}\cdot\frac{\sqrt{d-1}}{\zd(C)}\dist_{\max}(a,b).
\end{equation}
Now, consider $c_3\cdot \dist_{\max}(a^*,b^*)$.
Note that 
\begin{align}\label{eqn_lb_infty}
    c_3\cdot \dist_{\max}(a^*,b^*)&\ge  \max\left\{
\tfrac{\sqrt{d-1}}{a_z\cdot\zd(C)}\lVert x-x'\rVert_2,\ 
\frac{|z-z'|}{\zd(C)}
\right\}\geq\frac{\dist_{\max}(a,b)}{\zd(C)}.
\end{align}
Applying~\eqref{eqn_relation} to $a,b$ and to $a^*,b^*$ in~\eqref{eqn_lb_infty},
\begin{equation}\label{eqn_lb}
    \dist_{\mathbb{R}}(a^*,b^*)\ge \dist_{\max}(a^*,b^*)\geq \frac{1}{c_3}\cdot\frac{\dist_{\max}(a,b)}{\zd(C)}\ge \frac{1}{c_3\sqrt{2}}\cdot\frac{\dist_{\IR}(a,b)}{\zd(C)}.
\end{equation}
Due to~\eqref{eqn_ub} and \eqref{eqn_lb}, the claim follows.
\end{claimproof}

\begin{claim}\label{clm_constdistort2}
    \[ \frac{\dist_{\IR}(a,b)}{\zd(C)} = \Theta(\dist_{\IH}(a,b)).\]
\end{claim}

\begin{claimproof}
We first show that $\frac{\dist_{\IR}(a,b)}{2\zd(C)}=O(\dist_{\IH}(a,b))$.
For $c>0$ and $0 \le t \le r \le c$, we have
\[
\arsinh(r)
= \int_{0}^{r} \frac{1}{\sqrt{1 + t^2}}\,dt
\ge \int_{0}^{r} \frac{1}{\sqrt{1 + c^2}}\,dt
= \frac{r}{\sqrt{1 + c^2}}.
\]
Hence, we have
\[
\dist_{\IH}(a,b) \ge \frac{2}{\sqrt{1 + c^2}}\cdot\frac{\dist_{\IR}(a,b)}{2\sqrt{z z'}}.
\]
When $z,z' \le 2\zd(C)$ then $\sqrt{z z'} \le 2\zd(C)$. Therefore, we have
\begin{equation}\label{eqn_hyp_euc_lb}
    \frac{\dist_{\IR}(a,b)}{2\zd(C)}\leq \sqrt{1 + c^2}\cdot \dist_{\IH}(a,b) .
\end{equation}

We will now show that $\dist_{\IH}(a,b)=O\left( \frac{\dist_{\IR}(a,b)}{\zd(C)}\right)$.
Note that for all $x \ge 0$, we have $\arsinh(x) \le x$. Hence,
\[
\dist_{\IH}(a,b)
= 2 \arsinh\left(\frac{\dist_{\IR}(a,b)}{2 \sqrt{z z'}}\right)
\le 2 \cdot \frac{\dist_{\IR}(a,b)}{2 \sqrt{z z'}}
= \frac{\dist_{\IR}(a,b)}{\sqrt{z z'}}.
\]

Assuming $z, z' \ge \zd(C)$, we get $\sqrt{z z'} \ge \zd(C)$. As a result, we have
\begin{equation}\label{eqn_hyp_euc_ub}
    \dist_{\IH}(a,b)
\le \frac{\dist_{\IR}(a,b)}{\zd(C)}.
\end{equation}

Due to~\eqref{eqn_hyp_euc_lb} and~\eqref{eqn_hyp_euc_ub}, we have 
\[
\dist_{\IH}(a,b)
\le \frac{\dist_{\IR}(a,b)}{\zd(C)} 
\le 2 \sqrt{1 + c^2} \cdot \dist_{\IH}(a,b),
\]
as required.
\end{claimproof}

The claimed distance property now follows by combining \Cref{clm_constdistort1} with \Cref{clm_constdistort2}.
\end{proof}

\subsection{Portal placement} \label{Portal-placement}
For any pair of sibling cells in $\qt$, we now define a set of portals to be placed on their shared boundary.
The placement depends on the level $\ell$ of these siblings.

\subparagraph{Negative levels.}
When $\ell < 0$, we place the portals as in the Euclidean case~\cite{EucTSPJACM} with a slight modification described below. In the Euclidean setting, the optimal tour consists of $n$ line segments. Each segment crossing between two adjacent (Euclidean) quadtree cells must intersect their common facet.  In the hyperbolic setting, however, tour segments
can have endpoints in two adjacent negative-level cells but still avoid intersecting their common facet.
To address this, we slightly modify the height of the common facet by extending it in the $z$-direction.
This ensures that every geodesic intersects the extended facet; see \Cref{fig_arc}(ii).
Correspondingly, for a (hyperbolic or Euclidean) hyperplane $H$, we say here that a point $q \in H \cap \pi$ is a \EMPH{crossing} of the tour $\pi$, if there exists $i \in [n]$ such that $q$ lies on geodesic $p_i p_{i+1}$ and the endpoints $p_i$ and $p_{i+1}$ lie on different sides of $H$; see \Cref{fig_arc}(i).

\begin{figure}
    \centering
    \includegraphics[page=2,width=\textwidth]{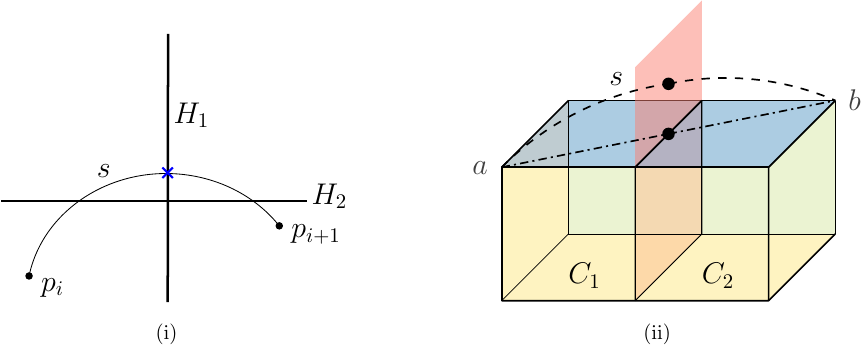}
    \caption{(i) The geodesic $s$ has a crossing with the hyperplane $H_1$, but not with the hyperplane $H_2$: although $s$ intersects $H_2$ twice, its endpoints lie on the same side.
    (ii) Two adjacent negative-level hybrid tree cells $C_1$ and $C_2$ share a common facet $F$. Extending $F$ upward by a factor of $4$ yields $F'$, which is intersected by any geodesic $s$ joining points $a\in C_1$ and $b\in C_2$.}
    \label{fig_arc}
\end{figure}
\begin{restatable}{lemma}{lemextendfacet}\label{lemma_extend_facet}
    Let $C_1$ and $C_2$ be two adjacent level-$\ell$ ($\ell < 0$) hybrid tree cells sharing a common facet $F$ of Euclidean height $\Delta=2^{2^\ell}-1$. If we extend this facet upward
    by a factor of $4$, then every geodesic with endpoints in $C_1$ and $C_2$ intersects the extended facet $F'$. 
\end{restatable}
\begin{proof}
    Let $\gamma$ be a geodesic with endpoints $a \in C_1$ and $b \in C_2$.  
    Recall that in the upper half–space model, a geodesic between two points with different horizontal coordinates is represented by a Euclidean circular arc orthogonal to the plane $z=0$. 
    Let $s=\arc{ab}$ be such an arc, with endpoints $a\in C_1$ and $b\in C_2$. Without loss of generality, assume that $a$ and $b$ are the vertices of the top facets of $C_1$ and $C_2$ opposite to the common facet $F$, so that they are diagonally opposite, and thus form the pair of points that are farthest apart across $F$. 
   Let the Euclidean coordinates of $a$ and $b$ are $(0,0,\cdots,0,\Delta+1)$ and $(2\alpha,\alpha,\cdots,\alpha, \Delta+1)$, respectively, where $\alpha=a_z \cdot 2^\ell/\sqrt{d-1}$ and $a_z\in[1,2)$.
    Notice that the Euclidean distance between $a$ and $b$ is as follows
    \begin{align*}
        \dist_{\IR}(a,b)< (2\cdot 2^\ell)\sqrt{1+3/(d-1)}< \frac{2(2^{2^\ell}-1)}{\ln 2}\cdot\sqrt{1+3/(d-1)}\tag{since $2^\ell\ln 2<2^{2^\ell}-1$}. 
    \end{align*}
    Since $\Delta=2^{2^{\ell}}-1$ and $d\geq 2$, we have $\dist_{\IR}(a,b)< 4\Delta/\ln 2$. Note that the Euclidean circular arc $s$ is shorter than a semicircle, i.e., it corresponds to an angle strictly less than $180^{\circ}$, the maximum (Euclidean) distance between segment $\overline{ab}$ and $s$ is strictly less than $\dist_{\IR}(a,b)/2$, which is strictly less than $2\Delta/\ln 2$.  
    Hence, the (Euclidean) height of $F'$ is $\Delta+2\Delta/\ln2$ and the height of $F$ was $\Delta$. Therefore, $F$ is extended by a factor of $1 + 2 / \ln 2$, which is strictly less than $4$, in the upward $z$-direction to obtain $F'$. Hence, the lemma follows.
\end{proof}

\subparagraph{Non-negative levels.}
We now consider the case $\ell \geq 0$. Here, even the previously modified Euclidean portal placement cannot be applied directly, as it would result in an exponential blow-up in the number of portals and thus fail to achieve the desired result. To overcome this, we introduce a different portal placement strategy, described below.
Note that a boundary facet shared between two sibling cells of level $\ell$ can be of two types:
\begin{description}
    \item[$\Tz$:] Facets that lie in horizontal Euclidean hyperplanes (hyperbolically these are horospheres). Each such facet is the top boundary of some binary tiling cell.
    \item[$\Tnz$:] Facets that lie in vertical hyperplanes. These are in fact $(d-1)$-dimensional cells, so each such facet corresponds to a cell in some $(d-1)$-dimensional hybrid tree.
\end{description}

Since $\Tz$ facets come from binary tiling cells, \Cref{lem_constdistort} means we have constant distortion compared to Euclidean space.
Consequently, we can simply construct a $(1/r)$-cover of the facet by placing a grid of $r^{d-1}\cdot d^{O(d)}$ equally-spaced portals.

For a $\Tnz$ facet $\sigma$, recall that $\sigma$ is a cell of some $(d-1)$-dimensional hybrid tree and can therefore be seen as a complete $2^{d-2}$-ary tree $T_\sigma$ of binary tiling cells.
We consider each cell $F$ in $T_\sigma$ individually and give a portal placement based on its (hop) distance $h$ to the root of $T_\sigma$.
Let $b = 2^{1-1/d}$.
If $b^h \leq r$, we place portals according to a $(b^h/r)$-cover in $F$.
Due to \Cref{lem_constdistort}, this means we place $d^{O(d)} \cdot (r/b^h)^{d-1}$ portals.
Otherwise, we do not place any portals.
Now, the total number of portals placed for a $\Tnz$ facet is bounded by
\[
    \sum_{h=0}^\infty 2^{h(d-2)} \cdot d^{O(d)} \cdot (r/b^h)^{d-1}
    = d^{O(d)} \cdot r^{d-1} \sum_{h=0}^\infty 2^{-h/d}= d^{O(d)} \cdot r^{d-1}.
\]

\begin{observation}\label{obs_cover}
  For any $\Tz$ or $\Tnz$ facet, we place $r^{d-1}\cdot d^{O(d)}$ portals.
\end{observation}

\section{Structure Theorem}\label{sec_structure}

Now we discuss and prove the structure theorem that allows us to develop a dynamic programming algorithm to solve the Hyperbolic TSP and Steiner Tree problem. We start by adapting the definitions for $r$-simple geometric graphs and $r$-simplifications from the Euclidean case \cite{EucTSPJACM} to our hyperbolic setting.

\begin{definition}[$r$-simple geometric graph]
Given a $(d-1)$-dimensional Euclidean box $F$, let $F^*$ denote $F$ with its $2^{d-1}$ corners removed. A geometric graph $\pi$ is called \EMPH{$r$-simple} if, for each facet $F$ shared by a pair of sibling cells in $\qt$,
\begin{itemize}
    \item if $F$ belongs to a non-negative level of the compressed hybrid tree, then the geometric graph $\pi$ crosses $F$ only through the portal points defined in $F$; or
    \item if $F$ belongs to a negative level of the compressed hybrid tree, then either $\pi$ crosses $F$ at any number of its corner points and exactly one point of $F^*$, or $\pi$ crosses $F$ only through the points from the uniform grid of size $g$ on $F$, for some $2^{d-1}\leq g\leq r^{2(d-1)}/{|\pi\cap F^*|}$.
\end{itemize}
\end{definition}
\begin{definition}[$r$-simplification]\label{def_eps_simple}
A geometric graph $\pi'$ is called an \EMPH{$r$-simplification} of a geometric graph $\pi$ if $\pi'$ is $r$-simple, and in each facet $F$ of a negative level where $|\pi \cap F^*| = 1$, the single non-corner crossing of $\pi'$ is the point in $\pi \cap F^*$.
\end{definition}

Having these definitions, we are now ready to state our main theorem.

\begin{restatable}[hyperbolic structure theorem]{theorem}{thmStrucMain}\label{thm:structure_theorem}
    Let $\textbf{a}$ be a random shift and let $\pi$ be a tour (or Steiner tree) of the point set $P\subset \IH^d$. Then for $r>0$, there is a tour (resp.\ Steiner tree) $\pi'$ of $P$ that is an $r$-simplification of $\pi$ such that
    \[\mathbb{E}_{\mathbf{a}}[\wt_{\IH}(\pi') - \wt_{\IH}(\pi)]\leq O(d^3)\cdot \frac{\wt_{\IH}(\pi)}{r}.\]
\end{restatable}

In order to prove \Cref{thm:structure_theorem}, we start by proving a bound on the patching cost of the negative levels of the compressed hybrid tree. Note that we state and prove everything for a tour, but every step works identically for a Steiner tree.

\begin{restatable}{lemma}{lemneglevelmod}\label{lem:neglevelmod}
    Let $\pi_C$ be the part of a tour $\pi$ inside a cell $C$ of the binary tiling. Then there exists a tour $\pi'_C$ in $C$ containing $C\cap P$ and connecting $a,b\in \bd C$ with a path if and only if $a,b\in \bd C$ are connected in $\pi_C$ such that $\pi'_C$ is either
    \begin{itemize}
        \item an $r$-simplification with respect to the negative levels and patched to the portals of~\cite{EucTSPJACM} or
        \item  $r$-simple with respect to the negative levels and patched to the portals of Arora~\cite{Arora98},
    \end{itemize}
     and $\pi'_C$ has expected length $(1+O(d^3/r))\wt_{\IH}(\pi_C)$. 
\end{restatable}
\begin{proof}

Take $\phi$ as given by \Cref{lem_constdistort}.
Let $\pi^*_C = \phi(\pi_C)$ denote the transformed image of the hyperbolic tour inside $C$, and let $P^*_C = \phi(P\cap C)$ be the corresponding transformed points. Note that $\pi^*_C$ is the Euclidean version of the original tour $\pi_C$, consisting of deformed hyperbolic geodesics (not Euclidean segments).

As described in \Cref{Portal-placement}, we modify the portal placement by Arora~\cite{Arora98} and Kisfaludi-Bak, Nederlof and W\k{e}grzycki~\cite{EucTSPJACM}, so that we have the same portal density in the extended facets. This modification alters the number of portals in each facet only by a constant factor and hence does not impact the asymptotic behavior of the patching cost analysis.
Apply the patching procedure of~\cite{Arora98} or~\cite{EucTSPJACM} to the intersection of the (deformed) hyperbolic geodesics and the randomly shifted Euclidean quadtree, excluding all hyperplanes that align with the boundaries of binary tiling cells. 

Let $\bar{\pi}_C$ denote the tour obtained by replacing each arc of $\pi^*_C$ with a Euclidean segment connecting the same endpoints, where possible endpoints are either transformed input points in $P^*_C$ or intersection points of $\pi^*_C$ with the transformed hyperplanes.
Consider a single hyperbolic segment $ab$ of $\pi_C$. Let $a^*=\phi(a),b^*=\phi(b)$ denote the transformed points. Then by~\Cref{lem_constdistort}, we have

\begin{equation}\label{eqn_distortion}
 \Omega\left(\dist_{\IH}(a,b)\right) \leq \dist_{\IR}(a^*,b^*) \leq O{\left(\sqrt{d}\cdot\dist_{\IH}(a,b)\right)}.
\end{equation}

Let $\Delta_\IR$ be the total expected Euclidean length of patching segments in the transformed space; by \Cref{obs:structproofEuc}, $\Delta_\IR \leq (d^{5/2}/r) \wt_{\IR}(\bar{\pi}_C)$.
By combining this with \eqref{eqn_distortion}, we have 
\[
    \Omega{\left(\frac{d^{5/2}}{r}\right)} \wt_{\IH}(\pi_C)
    \leq \Delta_\IR 
    \leq O{\left(\frac{d^3}{r}\right)} \wt_{\IH}(\pi_C).
\]
We now consider $\Delta_\IH$, the total expected hyperbolic length of these patching segments.
Note that $\Delta_\IH \leq O(\Delta_\IR)$ by \Cref{eqn_distortion}.
Hence, the (hyperbolic) cost of the patching is $\Delta_\IH = O(d^3/r)\wt_{\mathbb H}(\pi_C)$.
Adding these patching segments to the original tour, we conclude that there exists a tour $\pi'_C$ of expected length $(1+O(d^3/r))\wt_{\IH}(\pi_C)$, as required.
\end{proof}

Now, we prove our structure theorem.
\begin{proof}[Proof of~\Cref{thm:structure_theorem}.]
We describe a patching procedure for negative and non-negative levels of our shifted compressed hybrid tree separately.
\subparagraph{Patching negative levels.}
We first modify the optimum tour $\pi$ to get a tour $\pi_1$ that is slightly longer in expectation, and only crosses the shared boundaries of negative-level sibling cells through portals. In each binary tiling cell $C$, we apply \Cref{lem:neglevelmod} on $\pi_C=\pi\cap C$. As a result, we get a collection $\pi'_C$ of paths in $C$ whose expected length is $(1+O(d^3/r))\wt_{\IH}(\pi\cap C)$. Let $\pi_1$ denote the union of $\pi'_C$ over all binary tiling cells $C$.
Notice that $\pi_1\cap C$ connects two points of $\bd C$ if and only if the same points are connected by $\pi$, and $\pi_1$ covers the same set of input points; thus $\pi_1$ is a tour of $P$ of expected length $(1+O(d^3/r))\wt_{\IH}(\pi)$. 
    
\subparagraph{Patching non-negative levels.} 
We further modify the tour $\pi_1$ to get a tour $\pi_2$ that is slightly longer in expectation, and only crosses the cell boundaries between non-negative level siblings of the compressed hybrid tree through portals. This can be achieved as follows. For each crossing point $x$ of $\pi_1$ with a shared facet $F$ of cell siblings of level at least~$0$, we add two hyperbolic segments on each side of the facet connecting a copy of $x$ and the portal that is closest to $x$. (Technically, the segments are added at some infinitesimal distance from the facet.) This operation transfers all crossings to the nearest portals. Next, if some portal $v$ is used at least three times by the tour, then by \Cref{arora_patching} the tour can be simplified to go through the portal at most twice.

\subparagraph{Patching costs on non-negative levels.}

To bound the patching cost, we want to give deterministic bounds on the number of intersections between the tour $\pi$ and the hyperplanes used by the cells at level $\ell_{\min}$, then combine this with the expected cost for each intersection under random shifts ${\bf a}$.
We do this separately for the vertical hyperplanes and the horizontal ones, since their behavior and portal placement differ.

Let $\cal V$ be the set of all vertical hyperplanes obtained by extending
the vertical facets of the cells at level $\ell_{\min}$ in the unshifted hybrid tree. 
Let $\delta$ be the (maximum) diameter of all level $\ell_{\min}$ cells.
Since these cells are $\Omega(1/\sqrt{d})$-fat~\cite[Theorem 7]{Kisfaludi-BakW24}, we observe the following.
\begin{observation}\label{obs_mindist}
    After perturbation, the distance between two distinct points is $\Omega(\delta/\sqrt{d})$.
\end{observation}

Decompose $\pi$ into consecutive segments $\pi^1,\dots,\pi^{n}$.
Let $I$ denote the set of intersection points of $\pi$ and all hyperplanes in $\cal V$. To bound the patching cost of the vertical intersections of the tour $\pi$, we use a weighted portal analysis. We assign each $q \in I$ a weight equal to $1/z(q)$. 
To bound the expected patching cost of the $\Tnz$ facets, we start by bounding the total weight of intersections in $I$ as follows.

\begin{claim}\label{clm:intersected_tiles} 
$
\sum_{p\in I}\frac{1}{z(p)} \leq O{\left(\frac{d \sqrt d}{\delta}\cdot\wt_{\IH}(\pi)\right)}.
$
\end{claim}

\begin{claimproof}
We consider each of the $d-1$ horizontal axes separately; let $\cal{V}'$ be a set of hyperplanes perpendicular to one of these axes.
For each $j\in[n]$, let $I^j=\{p_1^j,\dots,p_{k_j}^j\}$ be the ordered set of intersection points of $\pi^j$ with $\cal V'$, and $I' = \bigcup_{j \in [n]} I^j$.
For $i\in[k_j-1]$, let $z_{i}^j = \min\{ z(p_i^j), z(p_{i+1}^j) \}$.
The Euclidean distance between adjacent hyperplanes from $\cal{V}'$ is given by the horobox width of level-$\ell_{\min}$ cells, which is $w \coloneq \Theta(\delta / \sqrt{d}) $ according to \cite[Lemma 6]{Kisfaludi-BakW24}.
By \cite[Observation 11]{Kisfaludi-BakW24}, the distance from a point $(w, 0, \dots, 0, z_{i}^j)$ to hyperplane $x_1=0$ is $\arsinh\frac{w}{z_{i}^j}$.
Since $\arsinh x = \Omega(x)$ for $x = O(1)$, $\wt_{\IH}\big(p_i^j p_{i+1}^j\big) =\Omega{\left(\frac{\delta}{z_{i}^j \sqrt d}\right)}$.
Summing this over $i\in[k_j-1]$ yields
\begin{align*}
     \wt_{\IH}(\pi^j) \geq \sum_{i\in[k_j-1]} \wt_{\IH}(p_i^j p_{i+1}^j)\geq\sum_{i\in[k_j-1]} \Omega\left(\frac{\delta}{z_{i}^j \sqrt{d}}\right).
\end{align*}
Since $1/(z_{i}^{j})=\max\{1/z(p_i^j),1/z(p_{i+1}^j)\}$,
we have 
$\sum_{i\in[k_j-1]}\frac{1}{z_{i}^j} \ge \frac{1}{2} \sum_{p\in I^j}\frac{1}{z(p)}$. Hence, $\sum_{p\in I^j} \Omega{\left(\frac{\delta}{\sqrt{d}\cdot z(p)}\right)}$ is a lower bound for $\wt_{\IH}(\pi^j)$. If $k_j \leq 1$ (that is $\pi^{j}$ intersects $\mathcal{V}'$ at most once), the lower bound follows from \Cref{obs_mindist}.
Summing over $j\in[n]$, we have
\[
  \wt_{\IH}(\pi)
  \geq \sum_{j\in [n]}\sum_{p\in I^j} \Omega\left(\frac{\delta}{\sqrt{d}\cdot z(p)}\right)
  = \frac{\delta}{\sqrt{d}}\sum_{p\in I'} \Omega\left(\frac{1}{z(p)}\right).
\]
Hence, $\sum_{p\in I'}\frac{1}{z(p)} \leq O{\left(\frac{\sqrt d}{\delta}\cdot \wt_{\IH}(\pi)\right)}$ and summing over all axes gives the claim.
\end{claimproof}

Next, we state a bound for the expected patching cost of vertical intersections.

\begin{claim}\label{lem_portaldist}
    Let $p \in I$. Then $\mathbb{E}_{\bf a}[\text{patching cost for }p] = O{\left(\frac{d \cdot \delta}{r \cdot z(p)}\right)}$.
\end{claim}
\begin{claimproof}
\begin{figure}
    \centering
    \hspace*{3cm}
    \includegraphics[page=2]{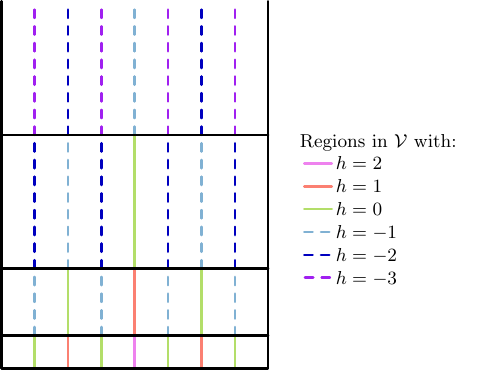}
    \caption{Distribution of the value $h$ used in \Cref{lem_portaldist} (here $\ell_{\min} = 0$). Note that the solid lines give the binary tiling.}
    \label{fig_h}
\end{figure}

    For any vertical hyperplane $V \in \mathcal V$, partition it into regions $(V_h)_{h \in \mathbb Z}$ with the horizontal Euclidean hyperplanes $z = 2^i / {\bf a}_z$ (for all $i \in \IZ$, unrelated to $h$) used by the shifted binary tiling.
    Here, $V_0$ denotes the highest region of $V$ that still separates two cells from the shifted binary tiling, and $V_{h+1}$ is the region directly below region $V_h$, as shown in \Cref{fig_h}.
    Note that for $h \geq 0$, any two regions $V_h$ and $V'_h$ (with $V \neq V'$ but the same~$h$) have the same portal placement:
    if $b^h \leq r$ the portals form a $(b^h / r)$-cover and otherwise there are no portals.
    For the given point~$p$, we now let the random variable $H$ be such that $p \in V_H$ for some $V \in \mathcal V$, and calculate the expected patching cost based on the behavior of $H$.

    Note that $H$ depends on the shift $\bf a$ as well as $z(p)$.
    Let us first determine the probability distribution of $H$ at the bottom row $z(p) = 1/{\bf a}_z$.
    Note that the number of shifts ${\bf a}_x$ giving $H=h$ is twice the number giving $H=h+1$.
    Thus, if we have $k$ shifts that give $H = \ell_{\min}$ then $\lfloor k \cdot 2^{\ell_{\min} - h - 1} \rfloor$ shifts give $H=h$.
    Seeing as we have $2k - 1$ shifts in total,
    \[ \mathbb{P}(H=h) = O(2^{\ell_{\min} - h}) = O(\delta / 2^h). \]
    For general $z(p)$, note that $H$ decreases by one every time $z(p)$ doubles and hence
    \[ \mathbb{P}(H=h) = O{\left(\delta / 2^{h + \log\frac{z(p)}{{\bf a}_z}}\right)} = O{\left(\frac{\delta}{2^h z(p)} \right)}, \]
    since ${\bf a}_z \in [1,2)$.
    For convenience, let $t = O(\delta / z(p))$ to simplify this to $\mathbb{P}(H=h) \leq t / 2^h$.
    
    When $H \geq 0$ and $b^H \leq r$, the patching cost will be at most $2b^H / r$.
    If $b^H > r$, we can still route through some other portal on the facet, which will be within distance $O(H - \log_b r)$.
    If $H < 0$, then the patching cost is zero.
    Therefore, we have
    \begin{align*}
        \mathbb{E}_{\bf a}[\text{patching cost for }p]
        &\leq \sum_{h = 0}^{\log_b r} 2b^h t / (r \cdot 2^h)
        + \sum_{h = \log_b r}^{\infty} O(h - \log_b r) \cdot t / 2^h \\
        &< \frac{2t}{r} \sum_{h = 0}^\infty (b / 2)^h
        + O(t / 2^{\log_b r}) \sum_{h = 0}^{\infty} h / 2^h \\
        &= \frac{2t}{r} \cdot \frac{1}{1 - b/2}
        + O(t / 2^{\log_b r}) \\
        &= \frac{2t}{r} \cdot \frac{1}{1 - 2^{-1/d}}
        + O(t / r^{\frac{1}{1-1/d}}) \\
        &= O(d t / r). \tag{since $(1+x)^r \le 1 + rx$ for reals $0 \le r \le 1$ and $x \ge -1$.}
    \end{align*}
    Hence, the claim follows.
\end{claimproof}

Consider the crossings of $\pi_1$ with $\Tnz$ facets.
By \Cref{clm:intersected_tiles} and \Cref{lem_portaldist}, we have
\begin{align*}
\sum_{p\in I}\mathbb{E}_{\mathbf a}[\text{patching cost for }p]
= O{\left(\frac{d\delta}{r}\right)} \sum_{p\in I}\frac{1}{z(p)}
= O{\left(\frac{d^2\sqrt d}{r}\right)} \wt_{\IH}(\pi_1).
\end{align*}

    Let $J$ denote the set of intersection points of $\pi$ and all hyperplanes in $\cal H$. Now the following claim bounds the cardinality of $J$, which will be will play an important role in bounding the number of crossings with $\Tz$ facets.
\begin{claim}\label{clm_card_J}
    Let $\cal H$ be the set of all horizontal Euclidean hyperplanes obtained by extending the horizontal facets of the cells at level $\ell_{\min}$ in the unshifted hybrid tree.  Then,
    \[|J|\leq O\left(\frac{\wt_{\IH}(\pi)\cdot \sqrt{d}}{\delta}\right).\]
\end{claim}
\begin{claimproof}
For each $i\in[n]$, let $J^j=\{p_1^j,\dots,p_{k_j}^j\}$ be the ordered set of intersection points of $\pi^j$ with $\cal H$. For every $j\in[k_j-1]$, the hyperbolic length of the segment $p_j^ip_{j+1}^i$ is
\begin{equation}\label{eq:local_1}
\wt_{\IH}\big(p_j^i p_{j+1}^i\big)
=\Omega\left(\frac{\delta}{\sqrt d}\right).
\end{equation}
Summing \eqref{eq:local_1} over $i\in[k_j-1]$ yields
\begin{align}\label{eqn_Ji}
     \wt_{\IH}(\pi^i)\geq \sum_{j\in[k_j-1]} \wt_{\IH}(p_j^i p_{j+1}^i)\geq\sum_{j\in[k_j-1]} \Omega\left(\frac{\delta}{\sqrt{d}}\right)\geq\Omega\left(\frac{\delta}{\sqrt{d}}\right)\cdot \frac{|J^j|}{2}.
\end{align}
For $k_j \le 1$, the lower bound $\wt_{\IH}(\pi^i) \ge \Omega\left({\delta}/{\sqrt{d}}\right){|J^j|}/{2}$ is a consequence of \Cref{obs_mindist}. Since the sets $J^1,\dots,J^{n}$ are disjoint and their union is $J$, we have $\wt_{\IH}(\pi)=\sum_{i\in[n]} \wt_{\IH}(\pi^i)$.
Summing~\eqref{eqn_Ji} over $i\in[n]$, we have
\[
  \wt_{\IH}(\pi)\geq\sum_{i\in [n]}\Omega\left(\frac{\delta}{\sqrt{d}}\right)\cdot \frac{|J^i|}{2}=\Omega\left(\frac{\delta}{\sqrt{d}}\right)\cdot \frac{|J|}{2}.
\]
After rearranging, we have
\[
|J| \leq O\left(\frac{\sqrt d}{\delta}\cdot \wt_{\IH}(\pi)\right), 
\] as required.
\end{claimproof}

To bound the number of crossings with $\Tz$ facets, we have the following claim.

\begin{restatable}{claim}{clmexpnegtopfacet}\label{clm_Exp_Neg_top_facet} Let $\Lambda$ be the number of intersections of $\pi$ with $\Tz$ facets of non-negative level cells. We have
$\mathbb{E}_{\bf a}[\Lambda]
= O(\wt_{\IH}(\pi)\sqrt{d})$.
\end{restatable}
\begin{claimproof}
For each $q\in J$, let $X_q$ be the random variable that indicates whether $q$ lies on a $\Tz$ facet of non-negative level (with respect to the random shift ${\bf a}$). By linearity of expectation,
\[
\mathbb{E}_{\bf a}[\Lambda]
= \mathbb{E}_{\bf a}\Big[\sum_{q\in J} X_q\Big]
= \sum_{q\in J}\mathbb{P}_{\bf a}(X_q=1).
\]

For any fixed $q$, under the random shift ${\bf a}$, the level of the $\Tz$ facet containing $q$ is uniformly distributed among the levels induced by $\cal H$. Since the fraction of non-negative levels is $2^{\ell_{\min}} = O(\delta)$, we have $\mathbb{P}_{\bf a}(X_q=1)=O(\delta)$ for all $q\in J$. Hence, we have 
\[
\mathbb{E}_{\bf a}\Big[\sum_{q\in J} X_q\Big]
= \sum_{q\in J} O(\delta)
= O(\delta)|J|.
\]

Due to~\Cref{clm_card_J}, we get $|J|\le O{\left(\frac{\wt_{\IH}(\pi)\sqrt{d}}{\delta}\right)}$. Thus, we obtain
\begin{align*}
    \mathbb{E}_{\bf a}[\Lambda]
&\le O(\delta)\cdot O{\left(\frac{\wt_{\IH}(\pi)\sqrt{d}}{\delta}\right)}\\
&= O(\wt_{\IH}(\pi)\sqrt{d}),
\end{align*}
as required.
\end{claimproof}

For $\Tz$ facets, the patching cost is always at most $2/r$ by construction, thus by \Cref{clm_Exp_Neg_top_facet} their expected total patching cost is $O{\left(\frac{\sqrt d}{r}\right)} \wt_{\IH}(\pi_1)$.
Combining the two contributions, the expected total patching cost is
$O{\left(\frac{d^2\sqrt d}{r}\right)} \wt_{\IH}(\pi_1)$.
Finally, using $\wt_{\IH}(\pi_1)=(1+O(1/r))\wt_{\IH}(\pi)$ to express the bound in terms of $\pi$, we obtain 
$\mathbb{E}_{\bf a}\big[\wt_{\IH}(\pi_2)-\wt_{\IH}(\pi_1)\big]
= O{\left(\frac{d^2\sqrt d}{r}\right)} \wt_{\IH}(\pi)$,
as required.
We conclude the proof by observing that $\pi'\coloneqq\pi_2$ is an $r$-simplification of $\pi$ and has the desired length.
\end{proof}

\section{Approximation Algorithm and its Time Complexity}\label{sec_DP}

In this section, we introduce polynomial time approximation schemes for TSP and Steiner Tree problem on a point set $P$ in $\IH^d$.
We design a dynamic programming algorithms that run in $2^{O(1/\eps^{d-1})} n (\log n)^{2d(d-1)}$ time for both the problems.

\subparagraph{Algorithm for the hyperbolic traveling salesman problem.}
At a high level, our algorithm follows the similar general framework as the PTAS of \cite{EucTSPJACM} for the Euclidean TSP in $\IR^d$. Let $r = \poly (1/\eps,1/d)$ be a parameter to be determined precisely later.

In \textit{Step~1}, we (implicitly) consider a hybrid tree data structure as introduced in \Cref{Compressed-Quadtree} for the smallest Euclidean bounding box containing the point set $P$ in $\IH^d$. Recall that leaf cells of the hybrid tree are located at level $\ell_{min}$. We relocate the points of $P$ in each cell of level $\ell_{\min}$ of the hybrid tree to its center. This step takes linear time.

In \textit{Step~2}, we apply a random shift $\bf a$ to the hybrid tree over the bounding box as described in \Cref{Shifting-Quadtree} and construct the compressed hybrid tree $\qt$ of the shifted hybrid tree in time $O(dn\log{n})$ using \Cref{lem:compressed-hybrid-tree-complexity}.

In \textit{Step~3}, as in \cite{EucTSPJACM}, we use the bottom-up dynamic programming to solve the \EMPH{$r$-multipath problem} in order to build our spanning path. For this, we place portals along the boundaries of the cells in the hybrid tree $\qt$, as described in \Cref{Portal-placement}. By \Cref{obs_cover}, the total number of portals on the boundaries of the non-negative level cells is $r^{d-1}\cdot d^{O(d)}$.

An input of the $r$-multipath problem consists of a nonempty cell $C$ in the shifted hybrid tree $\qt$, a subset $B$ of portals on the boundary of $C$, and a perfect matching $M$ on $B$. The $r$-multipath problem for a given instance $(C,B,M)$ finds a collection of $r$-simple paths $\mathcal{P}_{C,B,M}$ of minimum total length, where each matched pair of portals in $M$ gets a path with those portals as endpoints and each point $P \cap C$ is visited by some path from $\mathcal{P}_{C,B,M}$.

We employ the rank-based approach as in~\cite{EucTSPJACM} to solve the $r$-multipath problem. To describe the approach, we use notation similar to that in \cite{EucTSPJACM} and \cite{BergBKK23}. Let $C$ be the cell of the hybrid tree and let $B$ be an (even size) subset of portals on the boundary of $C$. Let $\mathcal{M}(B)$ denote the set of all matchings on $B$. We say a perfect matching $M_1$ \EMPH{fits} another perfect matching $M_2$ if their union is a Hamiltonian Cycle on $B$. For a subset $\mathcal{R}[B] \subseteq \mathcal{M}(B)$ and a fixed perfect matching $M$ we define

\[
    \text{opt}(M,\mathcal{R}[B]) \coloneqq \min \left\{ \text{weight of }\overline{M} : \overline{M} \in \mathcal{R}[B] \text{ and } \overline{M} \text{ fits } M \right\},
\]
where the weight of a perfect matching $M$ of $B$ is the total length of the solution to the $r$-multipath problem on $(C, B, M)$. Finally, we say that the set $\mathcal{R}[B] \subseteq \mathcal{M}(B)$ is \EMPH{representative} if for any matching $M$, we have $\text{opt}(M,\mathcal{R}[B]) = \text{opt}(M, \mathcal{M}(B))$. In the rank-based approach the input is a cell $C$ of the compressed hybrid tree. The output of the rank-based approach on $C$ is the representative sets $\mathcal{R}[B]$ for all (even size) possible subsets $B$ of portals on the boundary of $C$. As a base case, by \Cref{lem_pointsinleaf}, we have at most $2^{d-1}$ input points in a leaf cell of $\qt$. Therefore, we can find all the representative sets for each leaf cell in $O_d(1)$ time. Finally, the solution is obtained from the representative set $\mathcal{R}[B]$ corresponding to $B = \emptyset$ when $C$ is the root cell in the compressed hybrid tree. 

For negative levels, we can use the same algorithm for this rank-based approach as described in Remark~4.7 of \cite{EucTSPJACM}. For non-negative levels, our definition of $r$-simple is simpler so we can use the same approach but without handling single crossings separately. Note that we still always have $r^{d-1}\cdot d^{O(d)}$ portals and (by construction of the hybrid tree) at most $2^d$ children, so we get the same running time.

\subparagraph{Correctness.}
Let $\pi^*$ denote an optimal spanning tour of $P$. Recall that $\opt_\IH$ stands for the length of $\pi^*$. After Step~1, by
Claim~\ref{clm_perturb}, the modified tour $\pi$ obtained from $\pi^*$ by relocating the points of $P$, satisfies

\[
    \wt_\IH(\pi) \le (1+1/r)\opt_{\IH}.
\]

We now apply Theorem~\ref{thm:structure_theorem} to $\pi$. For the random shift
$\mathbf{a}$, the theorem guarantees the existence of an $r$-simplification
$\pi'$ of $\pi$ such that

\[
    \mathbb{E}_{\mathbf{a}}[\wt_\IH(\pi')]
    \le (1+O(d^3)/r)\,\wt_\IH(\pi).
\]

By construction of the dynamic program, the algorithm returns the optimal
$r$-simple tour with respect to the portals and grid induced by the same shift
$\mathbf{a}$. Since $\pi'$ is itself $r$-simple, the tour computed by the
algorithm has length at most $\wt(\pi')$. Therefore, the correctness follows by setting $r = \Theta(d^{-3}/\varepsilon)$.

\subparagraph{Time complexity.} As mentioned before, Step~1 and Step~2 of our algorithm take at most $O(dn\log{n})$ time. By Remark~4.7 of \cite{EucTSPJACM}, the running time of the rank-based approach to the solve the $r$-multipath problem (Step~3) is $2^{O(1/\eps^{d-1})}n(\log n)^{2d(d-1)}$. This concludes the proof of \Cref{thm_main} for the case of hyperbolic TSP.

\subparagraph{Algorithm for hyperbolic Steiner tree.} 

We now extend our hyperbolic TSP algorithm to solve the \emph{$r$-simple Steiner forest problem}, in order to design an approximation scheme for the Steiner tree problem (again $r = \Theta(d^{-3}/\eps)$). In contrast to the TSP setting, we track the connectivity requirements using partitions rather than matching. Thus, an input to the \EMPH{$r$-simple Steiner tree problem} consists of a nonempty cell $C$ in the shifted hybrid tree $\qt$, a subset $B$ of portals on the boundary of $C$, and a partition $M$ of $B$. We say a partition $M$ of $B$ is \EMPH{realized} by a forest if, for any $x, y \in B$, the vertices $x$ and $y$ lie in the same tree of the forest if and only if they belong to the same partition class of $M$. The $r$-simple Steiner forest problem for an instance $(C, B, M)$ finds an $r$-simple Steiner forest $\mathcal{F}_{C,B,M}$ of minimum total length that realizes the partition $M$ on $B$ and ensures that every point in $P \cap C$ is contained in some tree of $\mathcal{F}_{C,B,M}$.

In order to solve the $r$-simple Steiner forest problem using a dynamic program, we must introduce potential Steiner points into the leaf cells of the hybrid tree. We achieve this by constructing a  hyperbolic banyan. Note that we do not optimize the size of this banyan, as we merely require a banyan of size $\poly(n,1/\eps)$.

\begin{theorem}[Hyperbolic banyan] \label{thm:banyan}
    Let $d$ be a fixed constant, and let $P$ be a set of $n$ points in $\Hyp^d$. Then for any $\eps>0$ we can construct a $(1+\eps)$-banyan of $P$ of size $O_d(n^{d+3} \log{(1/\eps)}/ \eps^{(3d-1)/2})$ in $O_d(n^{d+3} \log{(n/\eps)}\log{(1/\eps)}/ \eps^{(3d-1)/2})$ time.
\end{theorem}

\begin{proof}

In order to find the $(1+\eps)$-banyan, we introduce a Steiner point placement $Q$. For distinct points $a, b \in P$, let $B_{a,b}$ be a hyperbolic ball of radius $\frac{3}{2}\dist_{\IH}(a,b)$ centered at the midpoint of the  hyperbolic segment $ab$. Let $C_{a,b}$ be $\conv(P \cap B_{a,b})$. Note that one can compute $\conv(P \cap B_{a,b})$ using a Euclidean algorithm after converting the points to the Cayley-Klein model, since there any Euclidean hyperplane is also a hyperbolic hyperplane and vice versa. Thus, we can do this in $O(n \log n + n^{\lfloor d/2 \rfloor})$ time \cite{convexhullanydim}.
Let $c$ be the maximum distance from any point in a hyperbolic $d$-simplex (i.e., a $d$-dimensional polytope defined by $d+1$ vertices) to one of its edges.

\begin{claim*}
    $c = O_d(1)$.
\end{claim*}
\begin{claimproof}
    For this, we show by induction that any point $p$ in a hyperbolic $d$-simplex $S$ must be at distance $O_d(1)$ from an edge.
    In the case $d=1$, the distance is $0$ because $1$-simplices are edges.
    Now, we prove the statement for $d$-simplices assuming it holds for $(d-1)$-simplices.
    Note that the hyperbolic $d$-simplex $S$ has ($d$-dimensional) volume $O_d(1)$~\cite{haagerup1981simplices}, while a hyperbolic ball of radius $\rho$ has volume $\Theta_d(1) \cdot \int_0^\rho \sinh^{d-1}(x)\,dx$~\cite{RIVIN20091297}, which is at least $\Theta_d(1) \cdot \int_0^\rho x^{d-1}\,dx = \Omega_d(\rho^d)$.
    Thus, the inscribed ball of $S$ has radius $O_d(1)$, meaning the distance from $p$ to the closest point $q$ on the boundary of $S$ is $O_d(1)$.
    Since $q$ lies on a boundary $(d-1)$-simplex, by the induction hypothesis $q$ has distance $O_d(1)$ to some edge of $S$ and $p$ has distance at most $O_d(1)$ to this edge.
\end{claimproof}

Let $\mu_{a,b} = \eps \dist_{\IH}(a,b) / (2n)$.
If $\mu_{a,b} \geq c$, triangulate $C_{a,b}$ with $d$-simplices, then define $Q_{a,b}$ by placing points at distance $\mu_{a,b}$ along each simplex edge.
Otherwise, we first define $\mathcal{T}_{a,b}$ to be the tiling given by the largest level of the (hyperbolic) quadtree of \cite{Kisfaludi-BakW24} that has cells of diameter at most $\mu_{a,b}$. Using a depth-first search, we find all cells of $\mathcal{T}_{a,b}$ that intersect $C_{a,b}$. We define $Q_{a,b}$ to be the set of centers of these cells. Note that finding the set $Q_{a,b}$ takes $O_d(|Q_{a,b}|)$ time, because $\mu_{a,b} = O_d(1)$ implies each cell has only $O_d(1)$ neighbors.

\begin{claim*}
    $Q_{a,b}$ is a $\mu_{a,b}$-cover of $C_{a,b}$ of size $O_d(n^{d+1}/\eps^d)$.
\end{claim*}

\begin{claimproof}
The case $\mu_{a,b} \geq c$ is simple:
$Q_{a,b}$ is a $\mu_{a,b}$-cover by the definition of $c$, while the size bound $|Q_{a,b}| = O_d(n^2 / \eps)$ follows because we are placing points at distance $\eps \dist_{\IH}(a,b) / (2n)$ along $O_d(n)$ edges of length $O(\dist_{\IH}(a,b))$.

In the case $\mu_{a,b} < c$, the fact that $Q_{a,b}$ is a $\mu_{a,b}$-cover follows by construction.
For the cardinality, we again triangulate $C_{a,b}$ into $n-d$ different $d$-simplices, then show by induction that any $d$-simplex intersects $O_d(n^d/\eps^d)$ cells of $\mathcal{T}_{a,b}$.
In the case $d=1$, any interval of length $O(\dist_{\IH}(a,b))$ naturally intersects $O(n / \eps)$ tiles.
We now prove the statement for a $d$-simplex $S$ assuming it holds for $(d-1)$-simplices.
First, the volume of $S$ is at most $O_d(\dist_\IH(a,b)^d)$ (this is tight for small distances, until the bound $O_d(1)$ becomes stronger).
Since cells of $\mathcal{T}_{a,b}$ are $\Omega_d(1)$-fat~\cite[Theorem~7]{Kisfaludi-BakW24}, they contain balls of radius $\Omega_d(\mu_{a,b})$ and thus have volume $\Omega_d(\mu_{a,b}^d)$.
Consequently, there can be at most $O_d(n^d / \eps^d)$ cells of $\mathcal{T}_{a,b}$ inside $S$.
Any cell that intersects $S$ but is not contained inside $S$ must intersect one of the $d+1$ boundary $(d-1)$-simplices of~$S$.

Consider such a boundary simplex $S'$ and the hyperplane $H$ it lies on.
Here, by the induction hypothesis, $O_d(n^{d-1}/\eps^{d-1})$ cells of the $(d-1)$-dimensional tiling $\mathcal{T}_{a,b}'$ of $H$ intersect $S'$.
We now need to convert this bound to one for cells of the $d$-dimensional tiling $\mathcal{T}_{a,b}$ of~$\IH^d$.
For this, we argue that any tile $T'$ of $\mathcal{T}_{a,b}'$ intersects $O_d(1)$ tiles of $\mathcal{T}_{a,b}$.
Note that any tile intersecting $T'$ must be within distance $O(\mu_{a,b})$ from it and thus inside a ball of radius $\rho = O(\mu_{a,b})$ (with slightly larger constant).
Since $\sinh x = O_d(x)$ for $x = O_d(1)$, such a ball has volume
\[
    \Theta_d(1) \cdot \int_0^\rho \sinh^{d-1}(x)\,dx
    = \Theta_d(1) \cdot \int_0^\rho x^{d-1}\,dx
    = O_d(\rho^{d}).
\]
As we already noted that cells have volume $\Omega_d(\mu_{a,b}^d)$, this means there can only be $O_d(1)$ of these inside any ball of radius $O(\rho)$.
Thus, there are $O_d(1)$ cells of $\mathcal{T}_{a,b}$ intersecting any cell of $\mathcal{T}_{a,b}'$.
Since by the induction hypothesis $O_d(n^{d-1}/\eps^{d-1})$ cells of $\mathcal{T}_{a,b}'$ intersect $S'$, now also $O_d(n^{d-1}/\eps^{d-1})$ cells of $\mathcal{T}_{a,b}$ intersect $S'$.
In total, $O_d(n^{d-1}/\eps^{d-1})$ cells of $\mathcal{T}_{a,b}$ now intersect the boundary of $S$, so there are still $O_d(n^d/\eps^d)$ tiles intersecting $S$.
\end{claimproof}

Finally, $Q$ is the union of sets $Q_{a,b}$ for all the distinct points $a,b \in P$.
We get a graph $G$ with $O_d(n^{d+3} \log{(1/\eps)} / \eps^{(3d-1)/2})$ edges by taking the Steiner $(1+\eps)$-spanner of \cite{Kisfaludi-BakvW25} on $P \cup Q$.
For any $P' \subseteq P$, now consider the diametric points $a,b \in P'$.
The Steiner tree for $P'$ will have weight $w \geq \dist_\IH(a,b)$.
Additionally, its at most $|P'|-1$ Steiner points lie inside $C_{a,b}$ and thus each is within distance $\mu_{a,b}$ from a point in $Q_{a,b}$.
We therefore add at most $2n\mu_{a,b} \leq \eps w$ to this Steiner tree's weight.
Thus, $G$ contains a $(1+\eps)$-approximate Steiner tree for any $P' \subset P$ and is thereby a $(1+\eps)$-banyan.
\end{proof}

We now use the rank-based approach to handle the $r$-simple Steiner forest problem, following the representative set framework in \cite{EucTSPJACM}. Note that compared to the rank-based approach of  hyperbolic TSP, we handle the base case differently. Our dynamic program starts by placing potential Steiner points in the leaf cells and the compressed cells of the compressed hybrid tree using banyans. Let $C$ be such a cell. Let $B$ be a subset of portals on the boundary of $C$. Let $Y$ be the union of $B$ and the points of $P \cap C$. By \Cref{lem_pointsinleaf}, $|P \cap C| \leq 2^{d-1}$ and thus $|Y| = r^{d-1}\cdot d^{O(d)}$. Let $G$ be a $(1+\eps)$-banyan on $Y$ guaranteed by \cref{thm:banyan}. Using $G$, we follow the Dreyfus-Wagner style dynamic program \cite{DreyfusW71} on $G$ to compute a minimum weight Steiner tree with terminals $Y$, as described in \cite{EucTSPJACM}. Each such query is stored at most once in the table and takes $2^{O(|Y|)} \poly(1/\eps)$ time to calculate (see \cite{DreyfusW71, EucTSPJACM}).

With this, the dynamic programming algorithm for solving the $r$-simple Steiner forest problem using the rank-based approach is the same as the hyperbolic TSP algorithm, except that it uses partitions rather than matchings on the portals along a cell’s boundary. To formalize this, we introduce the definition of representative sets for a cell $C$ with respect to the $r$-simple Steiner forest problem.
Let $B$ the set of portals on the boundary of $C$. For a subset $X$ of $B$, let $\mathcal{P}(X)$ denote the set of all partitions of $X$. We say $\mathcal{R}[X] \subseteq \mathcal{P}(X)$ is a \EMPH{representative} if for every Steiner forest $F$ such that $B_1,\dots,B_k$ are subsets of $B$ induced by the connected components of $F$ and $B_1 \cupdot \dots \cupdot B_k = X$, there exists a partition $M \in \mathcal{R}[X]$ such that the union of $F$ and a forest $F_M$ whose connected components correspond to $M$ gives a tree that spans $X$. Again, the output of the rank-based approach for the $r$-simple Steiner forest problem is the set of $\mathcal{R}[X]$ for all partitions $X$ of $B$. The correctness argument and time complexity of this algorithm are similar to those of the TSP algorithm.

\subparagraph{Derandomized algorithms.} We can derandomize our TSP and Steiner tree algorithm as follows.

\corrdeterministic*
\begin{proof}
    For this, we count the number of shifts given in \Cref{Shifting-Quadtree}, which is exactly the number of cells of $C_{\min}$ that lie in a level-$0$ (binary tiling) cell touching $z=1$.
    Let us first bound the number of these binary tiling cells.
    Here, recall that the binary tiling cells covered by a cell $C$ of level $\ell \geq 0$ form a complete $2^{d-1}$-ary tree of height $\ell$.
    Thus, there are $2^{(d-1)\ell}$ binary tiling cells in the bottom layer of this tree, which are the ones that will be used for shifting.
    Next, we note that $C$ is isometric to a cell of level $\log \ell$ of the hyperbolic quadtree of \cite{Kisfaludi-BakW24}, so their Theorem~7 says that $C$ has diameter $\Theta(\ell)$ and contains a ball of radius $\Omega(\ell/\sqrt{d})$.
    The specific cell $C$ we consider is the smallest cell containing $P$ (by definition the lexicographically minimal child of $B^*$), and therefore it now has level $\ell = O(\sqrt{d} \cdot \diam(P))$.
    Consequently, for fixed $d$, we have $2^{O(\diam(P))}$ cells of level $0$ on the bottom boundary.

    Next, we count the number of cells of $C_{\min}$ inside one binary tiling cell.
    Recall that $\ell_{\min}$ was defined such that the diameter of these cells is $\Theta(\eps\cdot\opt_{\IH} / n)$ (unless this would lead to $\ell_{\min} > 0$ in which case we set $\ell_{\min} = 0$ and we already know we have $2^{O(\diam(P))}$ shifts).
    Since the diameter of cells of level $\ell \leq 0$ is $\Theta(2^\ell)$ (by \Cref{lem_constdistort} or \cite[Theorem 7]{Kisfaludi-BakW24}), we have $\ell_{\min} = \Theta(\eps\cdot\opt_{\IH} / n))$.
    We now get two cases depending on $\opt_{\IH}$.
    If $\opt_{\IH} = \Omega(1)$, then $\ell_{\min} = \Omega(\log(\eps / n))$, so we get at most $2^{-d \cdot \ell_{\min}} = O{\left((n / \eps)^d\right)}$ cells.    
    If $\opt_{\IH} = O(1)$, then $\ell_{\max} = O(\log \opt_{\IH})$ (by \Cref{clm_diam}) so $\ell_{\max} - \ell_{\min} = O(\log(\eps / n))$.
    This again gives $O{\left((n / \eps)^d\right)}$ cells.
    
    Thus, the total number of shifts is $2^{O(\diam(P))} (n / \eps)^d$ and executing the steps of \Cref{thm_main} for each of these gives the claimed running time.
\end{proof}

\section{Conclusion} \label{conclusion}
In this paper we have obtained approximation schemes for $d$-dimensional hyperbolic TSP and hyperbolic Steiner tree with running time $2^{O(1/\eps^{d-1})}n\, \poly(\log n)$. Some interesting open questions remain for future work.
\begin{itemize}
    \item The reason why we do not yet match the running times of \cite{EucTSPJACM} is because the faster algorithms rely on light Euclidean spanners, while no light (Steiner) spanners exist for hyperbolic point sets as of yet. Are there Steiner spanners of lightness $\poly(1/\eps)$ in $\Hyp^d$? What about light banyans?
    \item Is it possible to use the techniques of \cite{Moemke25} to get a linear dependence on $n$?
    \item Is there a deterministic algorithm for hyperbolic TSP without dependence on the diameter of the input?
    \item It has been observed on several occasions that the complexity of optimization problems in hyperbolic space decreases as points are at larger distance from each other~\cite{BlasiusHKWW25,Kisfaludi-Bak21,Kisfaludi-BakvW25}. In line with this, one can embed any hyperbolic point set into an edge-weighted tree with additive error $O(\log n)$ \cite[Chapter 6]{Gromov1987} in time $O(n^2)$ \cite{gromovtree}, which already implies an $O(n^2)$ time TSP algorithm when $\opt_\IH = \Omega(n \log n / \eps)$.
    Can we already achieve milder dependence on $\eps$ for smaller $\opt_\IH$, or keep near-linear dependence on $n$?
\end{itemize}

\bibliography{hyperbolic}

@article{Kisfaludi-BakW24,
  author       = {S{\'{a}}ndor Kisfaludi{-}Bak and
                  Geert van Wordragen},
  title        = {A Quadtree, a {S}teiner Spanner, and Approximate Nearest Neighbours
                  in Hyperbolic Space},
  journal      = {Journal of Computational Geometry},
  volume       = {16},
  number       = {2},
  pages        = {145--174},
  year         = {2025},
  doi          = {10.20382/jocg.v16i2a5}
}

@inproceedings{VerbeekS14,
  author       = {Kevin Verbeek and
                  Subhash Suri},
  _editor       = {Siu{-}Wing Cheng and
                  Olivier Devillers},
  title        = {Metric Embedding, Hyperbolic Space, and Social Networks},
  booktitle    = {Proceedings of the 30th Symposium on Computational Geometry, {SoCG}},
  pages        = {501},
  publisher    = {{ACM}},
  year         = {2014},
  _url          = {https://doi.org/10.1145/2582112.2582139},
  doi          = {10.1145/2582112.2582139},
  timestamp    = {Tue, 07 May 2024 20:06:36 +0200},
  biburl       = {https://dblp.org/rec/conf/compgeom/VerbeekS14.bib},
  bibsource    = {dblp computer science bibliography, https://dblp.org}
}

@inproceedings{ShavittT04,
  author       = {Yuval Shavitt and
                  Tomer Tankel},
  title        = {On the Curvature of the Internet and its usage for Overlay Construction
                  and Distance Estimation},
  booktitle    = {Proceedings of the 23rd Annual Joint Conference
                  of the {IEEE} Computer and Communications Societies, {INFOCOM}},
  publisher    = {{IEEE}},
  year         = {2004},
  _url          = {https://doi.org/10.1109/INFCOM.2004.1354510},
  doi          = {10.1109/INFCOM.2004.1354510},
  timestamp    = {Wed, 16 Oct 2019 14:14:51 +0200},
  biburl       = {https://dblp.org/rec/conf/infocom/ShavittT04.bib},
  bibsource    = {dblp computer science bibliography, https://dblp.org}
}

@article{BorradaileKM15,
  author       = {Glencora Borradaile and
                  Philip N. Klein and
                  Claire Mathieu},
  title        = {A Polynomial-Time Approximation Scheme for {E}uclidean {S}teiner Forest},
  journal      = {{ACM} Trans. Algorithms},
  volume       = {11},
  number       = {3},
  pages        = {19:1--19:20},
  year         = {2015},
  _url          = {https://doi.org/10.1145/2629654},
  doi          = {10.1145/2629654},
  _timestamp    = {Tue, 06 Nov 2018 12:51:20 +0100},
  _biburl       = {https://dblp.org/rec/journals/talg/BorradaileKM15.bib},
  _bibsource    = {dblp computer science bibliography, https://dblp.org}
}

@article{Kisfaludi-BakvW25,
  author       = {S{\'{a}}ndor Kisfaludi{-}Bak and
                  Geert van Wordragen},
  title        = {Near-Optimal Dynamic {S}teiner Spanners for Constant-Curvature Spaces},
  journal      = {CoRR},
  _volume       = {abs/2509.01443},
  year         = {2025},
  _url          = {https://doi.org/10.48550/arXiv.2509.01443},
  _doi          = {10.48550/ARXIV.2509.01443},
  _eprinttype    = {arXiv},
  eprint       = {2509.01443},
  _timestamp    = {Wed, 08 Oct 2025 13:18:29 +0200},
  _biburl       = {https://dblp.org/rec/journals/corr/abs-2509-01443.bib},
  _bibsource    = {dblp computer science bibliography, https://dblp.org}
}

@article{Arora98,
  author       = {Sanjeev Arora},
  title        = {Polynomial Time Approximation Schemes for {E}uclidean Traveling Salesman
                  and other Geometric Problems},
  journal      = {J. {ACM}},
  volume       = {45},
  number       = {5},
  pages        = {753--782},
  year         = {1998},
  url          = {https://doi.org/10.1145/290179.290180},
  doi          = {10.1145/290179.290180},
  bibsource    = {dblp computer science bibliography, https://dblp.org}
}

@article{Mitchell99,
  author       = {Joseph S. B. Mitchell},
  title        = {Guillotine Subdivisions Approximate Polygonal Subdivisions: {A} Simple
                  Polynomial-Time Approximation Scheme for Geometric {TSP}, k-{MST}, and
                  Related Problems},
  journal      = {{SIAM} J. Comput.},
  volume       = {28},
  number       = {4},
  pages        = {1298--1309},
  year         = {1999},
  url          = {https://doi.org/10.1137/S0097539796309764},
  doi          = {10.1137/S0097539796309764},
  bibsource    = {dblp computer science bibliography, https://dblp.org}
}

@inproceedings{RaoS98,
  author       = {Satish Rao and
                  Warren D. Smith},
  title        = {Approximating Geometrical Graphs via "Spanners" and "Banyans"},
  booktitle    = {Proceedings of the 13th {ACM} Symposium on the Theory
                  of Computing, ({STOC})},
  pages        = {540--550},
  publisher    = {{ACM}},
  year         = {1998},
  doi          = {10.1145/276698.276868},
  bibsource    = {dblp computer science bibliography, https://dblp.org}
}

@inproceedings{BartalG13,
  author       = {Yair Bartal and
                  Lee{-}Ad Gottlieb},
  title        = {A Linear Time Approximation Scheme for {E}uclidean {TSP}},
  booktitle    = {Proceedings of the 54th {IEEE} Symposium on Foundations of Computer Science, ({FOCS})},
  pages        = {698--706},
  publisher    = {{IEEE}},
  year         = {2013},
  doi          = {10.1109/FOCS.2013.80},
  bibsource    = {dblp computer science bibliography, https://dblp.org}
}

@article{Moemke25,
        author={Tobias Mömke and Hang Zhou},
      title={A Linear Time Gap-{ETH}-Tight Approximation Scheme for {E}uclidean {TSP}}, 
      journal={CoRR},
      year={2025},
      eprint={2411.02585},
      _archivePrefix={arXiv},
      primaryClass={cs.DS},
}

@inproceedings{AroraRR98,
  author       = {Sanjeev Arora and
                  Prabhakar Raghavan and
                  Satish Rao},
  _editor       = {Jeffrey Scott Vitter},
  title        = {Approximation Schemes for {E}uclidean \emph{k}-Medians and Related Problems},
  booktitle    = {Proceedings of the 13th {ACM} Symposium on the Theory
                  of Computing, ({STOC})},
  pages        = {106--113},
  publisher    = {{ACM}},
  year         = {1998},
  _url          = {https://doi.org/10.1145/276698.276718},
  doi          = {10.1145/276698.276718},
  timestamp    = {Tue, 06 Nov 2018 11:07:04 +0100},
  biburl       = {https://dblp.org/rec/conf/stoc/AroraRR98.bib},
  bibsource    = {dblp computer science bibliography, https://dblp.org}
}

@inproceedings{KrauthgamerL06,
  author    = {Robert Krauthgamer and
               James R. Lee},
  title     = {Algorithms on negatively curved spaces},
  booktitle = {Proceedings of the 47th {IEEE} Symposium on Foundations of Computer Science, ({FOCS})},
  pages     = {119--132},
  publisher = {{IEEE}},
  year      = {2006},
  _url       = {https://doi.org/10.1109/FOCS.2006.9},
  doi       = {10.1109/FOCS.2006.9},
  timestamp = {Wed, 16 Oct 2019 14:14:54 +0200},
  biburl    = {https://dblp.org/rec/conf/focs/KrauthgamerL06.bib},
  bibsource = {dblp computer science bibliography, https://dblp.org}
}

@article{Kisfaludi-Bak21,
  author       = {S{\'{a}}ndor Kisfaludi{-}Bak},
  title        = {A quasi-polynomial algorithm for well-spaced hyperbolic {TSP}},
  journal      = {J. Comput. Geom.},
  volume       = {12},
  number       = {2},
  pages        = {25--54},
  year         = {2021},
  _url          = {https://doi.org/10.20382/jocg.v12i2a3},
  doi          = {10.20382/JOCG.V12I2A3},
  timestamp    = {Mon, 09 May 2022 16:20:13 +0200},
  biburl       = {https://dblp.org/rec/journals/jocg/Kisfaludi-Bak21.bib},
  bibsource    = {dblp computer science bibliography, https://dblp.org}
}

@inproceedings{HyperSteiner,
    author = {Alejandro García-Castellanos and Aniss Aiman Medbouhi and Giovanni Luca Marchetti and Erik J. Bekkers and Danica Kragic},
    title = {{HyperSteiner}: Computing Heuristic Hyperbolic {S}teiner Minimal Trees},
    booktitle = {Proceedings of the 2025 Symposium on Algorithm Engineering and Experiments (ALENEX)},
    chapter = {},
    pages = {194-208},
    doi = {10.1137/1.9781611978339.16},
    year = {2025},
    publisher = {{SIAM}},
}

@inproceedings{BlasiusHKWW25,
  author       = {Thomas Bl{\"{a}}sius and
                  Jean{-}Pierre von der Heydt and
                  S{\'{a}}ndor Kisfaludi{-}Bak and
                  Marcus Wilhelm and
                  Geert van Wordragen},
  _editor       = {Oswin Aichholzer and
                  Haitao Wang},
  title        = {Structure and Independence in Hyperbolic Uniform Disk Graphs},
  booktitle    = {Proceedings of the 41st International Symposium on Computational Geometry, ({SoCG})},
  series       = {LIPIcs},
  volume       = {332},
  pages        = {21:1--21:16},
  publisher    = {Schloss Dagstuhl - Leibniz-Zentrum f{\"{u}}r Informatik},
  year         = {2025},
  _url          = {https://doi.org/10.4230/LIPIcs.SoCG.2025.21},
  doi          = {10.4230/LIPICS.SOCG.2025.21},
  _timestamp    = {Tue, 05 Aug 2025 22:38:44 +0200},
  _biburl       = {https://dblp.org/rec/conf/compgeom/BlasiusHKWW25.bib},
  _bibsource    = {dblp computer science bibliography, https://dblp.org}
}

@article{EucTSPJACM,
author = {Kisfaludi-Bak, S\'{a}ndor and Nederlof, Jesper and W\k{e}grzycki, Karol},
title = {A Gap-{ETH}-Tight Approximation Scheme for {E}uclidean {TSP}},
year = {2025},
publisher = {Association for Computing Machinery},
_address = {New York, NY, USA},
issn = {0004-5411},
url = {https://doi.org/10.1145/3766548},
doi = {10.1145/3766548},
journal = {J. ACM},
_month = sep,
keywords = {Euclidean TSP, Euclidean Steiner Tree, Efficient Polynomial Time Approximation Scheme}
}

@article{Dinur16,
  author    = {Irit Dinur},
  title     = {Mildly exponential reduction from gap 3{SAT} to polynomial-gap label-cover},
  journal   = {Electron. Colloquium Comput. Complex.},
  volume    = {23},
  pages     = {128},
  year      = {2016},
  url       = {http://eccc.hpi-web.de/report/2016/128},
  bibsource = {dblp computer science bibliography, https://dblp.org}
}

@inproceedings{ManurangsiR17,
  author    = {Pasin Manurangsi and
               Prasad Raghavendra},
  _editor    = {Ioannis Chatzigiannakis and
               Piotr Indyk and
               Fabian Kuhn and
               Anca Muscholl},
  title     = {{A} {B}irthday {R}epetition {T}heorem and {C}omplexity of {A}pproximating {D}ense
	  {CSP}s},
  booktitle = {Proceedings of the 44th International Colloquium on Automata, Languages, and Programming,
               ({ICALP})},
  series    = {LIPIcs},
  volume    = {80},
  pages     = {78:1--78:15},
  publisher = {Schloss Dagstuhl - Leibniz-Zentrum f{\"{u}}r Informatik},
  year      = {2017},
  _url       = {https://doi.org/10.4230/LIPIcs.ICALP.2017.78},
  doi       = {10.4230/LIPIcs.ICALP.2017.78},
  timestamp = {Tue, 11 Feb 2020 15:52:14 +0100},
  biburl    = {https://dblp.org/rec/conf/icalp/ManurangsiR17.bib},
  bibsource = {dblp computer science bibliography, https://dblp.org}
}

@article{BartalGK16Siam,
  author       = {Yair Bartal and
                  Lee{-}Ad Gottlieb and
                  Robert Krauthgamer},
  title        = {The Traveling Salesman Problem: Low-Dimensionality Implies a Polynomial
                  Time Approximation Scheme},
  journal      = {{SIAM} J. Comput.},
  volume       = {45},
  number       = {4},
  pages        = {1563--1581},
  year         = {2016},
  url          = {https://doi.org/10.1137/130913328},
  doi          = {10.1137/130913328},
  bibsource    = {dblp computer science bibliography, https://dblp.org}
}

@inproceedings{Talwar04,
  author       = {Kunal Talwar},
  _editor       = {L{\'{a}}szl{\'{o}} Babai},
  title        = {Bypassing the embedding: algorithms for low dimensional metrics},
  booktitle    = {Proceedings of the 36th {ACM} Symposium on Theory of Computing, ({STOC})},
  pages        = {281--290},
  publisher    = {{ACM}},
  year         = {2004},
  url          = {https://doi.org/10.1145/1007352.1007399},
  doi          = {10.1145/1007352.1007399},
  timestamp    = {Tue, 06 Nov 2018 11:07:04 +0100},
  biburl       = {https://dblp.org/rec/conf/stoc/Talwar04.bib},
  bibsource    = {dblp computer science bibliography, https://dblp.org}
}

@article{BergBKK23,
  author       = {Mark de Berg and
                  Hans L. Bodlaender and
                  S{\'{a}}ndor Kisfaludi{-}Bak and
                  Sudeshna Kolay},
  title        = {{A}n {ETH}-{T}ight {E}xact {A}lgorithm for {E}uclidean {TSP}},
  journal      = {{SIAM} J. Comput.},
  volume       = {52},
  number       = {3},
  pages        = {740--760},
  year         = {2023},
  _url          = {https://doi.org/10.1137/22m1469122},
  doi          = {10.1137/22M1469122},
  timestamp    = {Sun, 12 Nov 2023 02:19:20 +0100},
  biburl       = {https://dblp.org/rec/journals/siamcomp/BergBKK23.bib},
  bibsource    = {dblp computer science bibliography, https://dblp.org}
}

@inproceedings{BartalG21,
  author       = {Yair Bartal and
                  Lee{-}Ad Gottlieb},
  _editor       = {Samir Khuller and
                  Virginia Vassilevska Williams},
  title        = {{N}ear-linear time approximation schemes for {S}teiner tree and forest in low-dimensional spaces},
  booktitle    = {Proceedings of the 53rd {ACM} {SIGACT} Symposium on Theory of Computing, ({STOC})},
  pages        = {1028--1041},
  publisher    = {{ACM}},
  year         = {2021},
  _url          = {https://doi.org/10.1145/3406325.3451063},
  doi          = {10.1145/3406325.3451063},
  _timestamp    = {Tue, 22 Jun 2021 20:03:56 +0200},
  _biburl       = {https://dblp.org/rec/conf/stoc/BartalG21.bib},
  _bibsource    = {dblp computer science bibliography, https://dblp.org}
}

@inproceedings{GaneaBH18,
  author    = {Octavian{-}Eugen Ganea and Gary B{\'{e}}cigneul and Thomas Hofmann},
  title     = {Hyperbolic Neural Networks},
  booktitle = {Proceedings of the Advances in Neural Information Processing Systems, (NeurIPS)},
  pages     = {5350--5360},
  year      = {2018},
  url       = {https://proceedings.neurips.cc/paper/2018/hash/dbab2adc8f9d078009ee3fa810bea142-Abstract.html}
}

@inproceedings{NickelK18,
  author    = {Maximilian Nickel and Douwe Kiela},
  title     = {Learning Continuous Hierarchies in the {L}orentz Model of Hyperbolic Geometry},
  booktitle = {Proceedings of the 35th International Conference on Machine Learning ({ICML})},
  volume    = {80},
  series    = {Proceedings of Machine Learning Research},
  pages     = {3776--3785},
  publisher = {PMLR},
  year      = {2018},
  _url       = {http://proceedings.mlr.press/v80/nickel18a.html}
}

@article{BlasiusFKMPW22,
  author    = {Thomas Bl{\"{a}}sius and Tobias Friedrich and Maximilian Katzmann and Ulrich Meyer and Manuel Penschuck and Christopher Weyand},
  title     = {Efficiently Generating Geometric Inhomogeneous and Hyperbolic Random Graphs},
  journal   = {Network Science},
  volume    = {10},
  number    = {4},
  pages     = {361--380},
  year      = {2022},
  doi       = {10.1017/nws.2022.32},
  _url       = {https://doi.org/10.1017/nws.2022.32}
}

@inproceedings{BlasiusFK16,
  author    = {Thomas Bl{\"{a}}sius and Tobias Friedrich and Anton Krohmer},
  title     = {Hyperbolic Random Graphs: Separators and Treewidth},
  booktitle = {Proceedings of the 24th European Symposium on Algorithms, ({ESA})},
  volume    = {57},
  series    = {LIPIcs},
  pages     = {15:1--15:16},
  publisher = {Schloss Dagstuhl -- Leibniz-Zentrum f{\"{u}}r Informatik},
  year      = {2016},
  doi       = {10.4230/LIPIcs.ESA.2016.15},
  _url       = {https://doi.org/10.4230/LIPIcs.ESA.2016.15}
}

@inproceedings{Filtser19,
  author       = {Arnold Filtser},
  _editor       = {Dimitris Achlioptas and
                  L{\'{a}}szl{\'{o}} A. V{\'{e}}gh},
  title        = {On Strong Diameter Padded Decompositions},
  booktitle    = {Proceedings of the Approximation, Randomization, and Combinatorial Optimization. Algorithms
                  and Techniques, ({APPROX/RANDOM})},
  series       = {LIPIcs},
  volume       = {145},
  pages        = {6:1--6:21},
  publisher    = {Schloss Dagstuhl - Leibniz-Zentrum f{\"{u}}r Informatik},
  year         = {2019},
  _url          = {https://doi.org/10.4230/LIPIcs.APPROX-RANDOM.2019.6},
  doi          = {10.4230/LIPICS.APPROX-RANDOM.2019.6},
  timestamp    = {Wed, 21 Aug 2024 22:46:00 +0200},
  biburl       = {https://dblp.org/rec/conf/approx/Filtser19.bib},
  bibsource    = {dblp computer science bibliography, https://dblp.org}
}

@article{ChanGMZ16,
  author       = {T.{-}H. Hubert Chan and
                  Anupam Gupta and
                  Bruce M. Maggs and
                  Shuheng Zhou},
  title        = {On Hierarchical Routing in Doubling Metrics},
  journal      = {{ACM} Trans. Algorithms},
  volume       = {12},
  number       = {4},
  pages        = {55:1--55:22},
  year         = {2016},
  _url          = {https://doi.org/10.1145/2915183},
  doi          = {10.1145/2915183},
  timestamp    = {Mon, 17 Feb 2025 21:53:32 +0100},
  biburl       = {https://dblp.org/rec/journals/talg/ChanGMZ16.bib},
  bibsource    = {dblp computer science bibliography, https://dblp.org}
}

@inproceedings{GuptaKL03,
  author       = {Anupam Gupta and
                  Robert Krauthgamer and
                  James R. Lee},
  title        = {Bounded Geometries, Fractals, and Low-Distortion Embeddings},
  booktitle    = {Proceedings of the 44th Symposium on Foundations of Computer Science, ({FOCS})},
  pages        = {534--543},
  publisher    = {{IEEE}},
  year         = {2003},
  _url          = {https://doi.org/10.1109/SFCS.2003.1238226},
  doi          = {10.1109/SFCS.2003.1238226},
  timestamp    = {Tue, 08 Jul 2025 16:41:41 +0200},
  biburl       = {https://dblp.org/rec/conf/focs/GuptaKL03.bib},
  bibsource    = {dblp computer science bibliography, https://dblp.org}
}

@article{MnasriJS15,
  title={Critical phenomena in hyperbolic space},
  author={Karim Mnasri and Bhilahari Jeevanesan and J{\"o}rg Schmalian},
  journal={Physical Review B},
  year={2015},
  volume={92},
  pages={134423},
  url={https://api.semanticscholar.org/CorpusID:118004293}
}

@article{BringmannKL19,
  author    = {Karl Bringmann and Ralph Keusch and Johannes Lengler},
  title     = {Geometric Inhomogeneous Random Graphs},
  journal   = {Theor. Comput. Sci.},
  volume    = {760},
  pages     = {35--54},
  year      = {2019},
  doi       = {10.1016/j.tcs.2018.08.014},
  url       = {https://doi.org/10.1016/j.tcs.2018.08.014}
}

@article{FriedrichK18,
  author    = {Tobias Friedrich and Anton Krohmer},
  title     = {On the Diameter of Hyperbolic Random Graphs},
  journal   = {{SIAM} J. Discret. Math.},
  volume    = {32},
  number    = {2},
  pages     = {1314--1334},
  year      = {2018},
  doi       = {10.1137/17M1123961},
  url       = {https://doi.org/10.1137/17M1123961}
}

@article{gromovtree,
      title={Gromov's Approximating Tree and the All-Pairs Bottleneck Paths Problem}, 
      author={Anders Cornect and Eduardo Martínez-Pedroza},
      journal={CoRR},
      year={2025},
      eprint={2408.05338},
      _archivePrefix={arXiv},
      primaryClass={cs.CG},
      _url={https://arxiv.org/abs/2408.05338}, 
}

@Inbook{Gromov1987,
author="Gromov, M.",
_editor="Gersten, S. M.",
title="Hyperbolic Groups",
bookTitle="Essays in Group Theory",
year="1987",
publisher="Springer",
_address="New York, NY",
pages="75--263",
isbn="978-1-4613-9586-7",
doi="10.1007/978-1-4613-9586-7_3",
url="https://doi.org/10.1007/978-1-4613-9586-7_3"
}

@article{Boroczky74,
  author    = {Károly Böröczky},
  title = {G\"ombkit\"olt\'esek \'alland\'o {g\"orb\"ulet\H{u}} terekben {I}},
  journal   = {Matematikai Lapok (in Hungarian)},
  volume    = {25},
  number    = {3-4},
  pages     = {265–306},
  year      = {1974},
  _doi       = {10./17M1123961},
  _url       = {https://doi.org/10.1137/17M1123961}
}

@article{CannonFKP97,
  title={Hyperbolic geometry},
  author={Cannon, James W and Floyd, William J and Kenyon, Richard and Parry, Walter R and others},
  journal={Flavors of geometry},
  volume={31},
  number={59-115},
  pages={2},
  year={1997}
}

@article{ungar2013einstein,
      title={Einstein's Special Relativity: The Hyperbolic Geometric Viewpoint}, 
      author={Abraham A. Ungar},
       journal={CoRR},
      year={2013},
      eprint={1302.6961},
      _archivePrefix={arXiv},
      primaryClass={math-ph},
      _url={https://arxiv.org/abs/1302.6961}, 
}

@article{haagerup1981simplices,
 author = {Haagerup, Uffe and Munkholm, Hans J.},
 title = {Simplices of maximal volume in hyperbolic n-space},
 fjournal = {Acta Mathematica},
 journal = {Acta Math.},
 issn = {0001-5962},
 volume = {147},
 pages = {1--11},
 year = {1981},
 language = {English},
 doi = {10.1007/BF02392865},
 zbMATH = {3776387},
 Zbl = {0493.51016}
}

@article{convexhullanydim,
  author       = {Bernard Chazelle},
  title        = {An Optimal Convex Hull Algorithm in Any Fixed Dimension},
  journal      = {Discret. Comput. Geom.},
  volume       = {10},
  pages        = {377--409},
  year         = {1993},
  url          = {https://doi.org/10.1007/BF02573985},
  doi          = {10.1007/BF02573985},
  bibsource    = {dblp computer science bibliography, https://dblp.org}
}

@article{CohenAddadFS21,
  author  = {Vincent Cohen-Addad and Andreas Emil Feldmann and David Saulpic},
  title   = {Near-linear Time Approximation Schemes for Clustering in Doubling Metrics},
  journal = {J. {ACM}},
  volume  = {68},
  number  = {6},
  pages   = {44:1--44:34},
  year    = {2021},
  doi     = {10.1145/3477541}
}

@inproceedings{CohenAddad20,
  author    = {Vincent Cohen-Addad},
  title     = {Approximation Schemes for Capacitated Clustering in Doubling Metrics},
  booktitle = {Proceedings of the 31st ACM-SIAM Symposium on Discrete Algorithms (SODA)},
  pages     = {2241--2259},
  year      = {2020},
  publisher = {SIAM},
  doi       = {10.1137/1.9781611975994.138}
}

@article{WuFW24,
  author  = {Di Wu and Qilong Feng and Jianxin Wang},
  title   = {Approximation Algorithms for Fair k-Median Problem without Fairness Violation},
  journal = {Theor. Comput. Sci.},
  volume  = {985},
  pages   = {114332},
  year    = {2024},
  issn    = {0304-3975},
  doi     = {10.1016/j.tcs.2023.114332}
}

@inproceedings{CzumajJKV22,
  author    = {Artur Czumaj and Shaofeng H.-C. Jiang and Robert Krauthgamer and Pavel Vesel{\'y}},
  title     = {Streaming Algorithms for Geometric Steiner Forest},
  booktitle = {Proceedings of the 49th International Colloquium on Automata, Languages, and Programming, ({ICALP})},
  series    = {LIPIcs},
  volume    = {229},
  pages     = {47:1--47:20},
  year      = {2022},
  publisher = {Schloss Dagstuhl -- Leibniz-Zentrum f{\"u}r Informatik},
  doi       = {10.4230/LIPIcs.ICALP.2022.47}
}

@article{FakcharoenpholW25,
  author  = {Jittat Fakcharoenphol and Nonthaphat Wongwattanakij},
  title   = {A {PTAS} for k-hop {MST} on the Euclidean Plane: Improving Dependency on k},
  journal = {Information Processing Letters},
  volume  = {190},
  pages   = {106581},
  year    = {2025},
  doi     = {10.1016/j.ipl.2025.106581}
}

@article{CohenAddadKM19,
  author  = {Vincent Cohen-Addad and Philip N. Klein and Claire Mathieu},
  title   = {Local Search Yields Approximation Schemes for k-Means and k-Median in Euclidean and Minor-Free Metrics},
  journal = {SIAM Journal on Computing},
  volume  = {48},
  number  = {2},
  pages   = {644--667},
  year    = {2019},
  doi     = {10.1137/17M112717X},
  note    = {Preliminary version in FOCS 2016}
}

@article{DreyfusW71,
  author       = {Stuart E. Dreyfus and
                  Robert A. Wagner},
  title        = {The steiner problem in graphs},
  journal      = {Networks},
  volume       = {1},
  number       = {3},
  pages        = {195--207},
  year         = {1971},
  _url          = {https://doi.org/10.1002/net.3230010302},
  doi          = {10.1002/NET.3230010302},
  timestamp    = {Thu, 19 Oct 2023 15:34:51 +0200},
  biburl       = {https://dblp.org/rec/journals/networks/DreyfusW71.bib},
  bibsource    = {dblp computer science bibliography, https://dblp.org}
}

@article{RIVIN20091297,
    title = {Asymptotics of convex sets in {E}uclidean and hyperbolic spaces},
    journal = {Advances in Mathematics},
    volume = {220},
    number = {4},
    pages = {1297-1315},
    year = {2009},
    issn = {0001-8708},
    doi = {https://doi.org/10.1016/j.aim.2008.11.014},
    author = {Igor Rivin}
}

@book{iversen1992hyperbolic,
  title={Hyperbolic Geometry},
  author={Iversen, Birger},
  volume={25},
  year={1992},
  series={London Mathematical Society Student Texts},
  publisher={Cambridge University Press},
doi={10.1017/CBO9780511569333}
}

@book{thurston97three,
  title={Three-Dimensional Geometry and Topology, Volume 1},
  journal={Princeton Mathematical Series, Vol. 35},
  author={Thurston, William P.},
  year={1997},
  publisher={Princeton University Press},
doi = {10.1515/9781400865321},
isbn = {9781400865321}
}

@book{benedetti1992lectures,
  title={Lectures on Hyperbolic Geometry},
  author={Benedetti, Riccardo and Petronio, Carlo},
  year={1992},
  publisher={Springer Science \& Business Media},
  doi={10.1007/978-3-642-58158-8}
}

@incollection{Aluru04,
  author    = {Srinivas Aluru},
  *editor    = {Dinesh P. Mehta and
               Sartaj Sahni},
  title     = {Quadtrees and Octrees},
  booktitle = {Handbook of Data Structures and Applications, 2nd edition},
  publisher = {Chapman and Hall/CRC},
  year      = {2018},
  chapter   = {20},
  pages     = {309--328},
  doi       = {10.1201/9781315119335}
}

@inproceedings{BuschCFHHR23,
  author       = {Costas Busch and
                  Da Qi Chen and
                  Arnold Filtser and
                  Daniel Hathcock and
                  D. Ellis Hershkowitz and
                  Rajmohan Rajaraman},
  title        = {One Tree to Rule Them All: Poly-Logarithmic Universal Steiner Tree},
  booktitle    = {Proceedings of the 64th {IEEE} Symposium on Foundations of Computer Science, ({FOCS})},
  pages        = {60--76},
  publisher    = {{IEEE}},
  year         = {2023},
  _url          = {https://doi.org/10.1109/FOCS57990.2023.00012},
  doi          = {10.1109/FOCS57990.2023.00012},
  timestamp    = {Tue, 08 Jul 2025 16:39:40 +0200},
  biburl       = {https://dblp.org/rec/conf/focs/BuschCFHHR23.bib},
  bibsource    = {dblp computer science bibliography, https://dblp.org}
}

@inproceedings{La025,
  author       = {An La and
                  Hung Le},
  _editor       = {Michal Kouck{\'{y}} and
                  Nikhil Bansal},
  title        = {Dynamic Locality Sensitive Orderings in Doubling Metrics},
  booktitle    = {Proceedings of the 57th {ACM} Symposium on Theory of Computing,
                  ({STOC})},
  pages        = {1178--1189},
  publisher    = {{ACM}},
  year         = {2025},
  _url          = {https://doi.org/10.1145/3717823.3718209},
  doi          = {10.1145/3717823.3718209},
  timestamp    = {Sun, 02 Nov 2025 21:27:35 +0100},
  biburl       = {https://dblp.org/rec/conf/stoc/La025.bib},
  bibsource    = {dblp computer science bibliography, https://dblp.org}
}
\end{document}